\documentclass[a4paper,oneside,11pt]{amsart} 
\usepackage{amsmath, amsfonts, amsthm, amssymb, amscd}
\usepackage[mathcal]{euscript} 
\usepackage{dsfont, pxfonts, verbatim}
\usepackage{mathrsfs}  
\newtheorem{theorem}{Theorem}[section]
\newtheorem{proposition}[theorem]{Proposition} 
\newtheorem{lemma}[theorem]{Lemma}
 
\newtheorem{definition}[theorem]{Definition} 
\newtheorem{remark}[theorem]{Remark}  
\numberwithin{equation}{section} 
\usepackage{a4wide}
\newcommand \Acal {\mathcal A}
\newcommand \Ecal {\mathcal E}
\newcommand \tildeU {\widetilde U}
\newcommand \tildef {\widetilde f}
\newcommand \la \langle
\newcommand \ra \rangle

\newcommand \del	    \partial  
\newcommand \auth 	{\textsc}   
\newcommand \Mcal 	{\mathcal M}

\newcommand \Escr 	{\mathscr E}  
\newcommand \Gscr 	{\mathscr G}

\newcommand \Pscr 	{\mathscr P}  
\newcommand \Qscr 	{\mathcal Q}  
\newcommand \Gcal 	{\mathcal G}
\newcommand \Qcal 	{\mathcal Q}  
\newcommand \Fcal       {\mathcal F}  
\newcommand \Escrhat 	{\widehat {\mathscr E}}

\newcommand \RR 		{\mathbb R}   
\newcommand \eps 	\epsilon  
\newcommand \be 		{\begin{equation}}
\newcommand \ee 		{\end{equation}} 
\newcommand{\aver}[1]{\left\la\widetilde{#1} \, \right\ra}

\let\oldmarginpar\marginpar
\renewcommand\marginpar[1]{\-\oldmarginpar[\raggedleft\footnotesize #1]%
{\raggedright\footnotesize #1}}

\begin{document}
 
\title[Future asymptotics of polarized $T^2$-symmetric spacetimes]{Future asymptotics and geodesic completeness
\\
of polarized $T^2$-symmetric spacetimes}  
\author{ 
Philippe G. L{\scriptsize e}Floch}
\address{Laboratoire Jacques-Louis Lions \& Centre National de la Recherche Scientifique, 
Universit\'e Pierre et Marie Curie (Paris 6), 4 Place Jussieu, 75252 Paris, France. 
Blog: {\sl philippelefloch.org.}}
\email{contact@philippelefloch.org}

\author{Jacques Smulevici}
\address{Laboratoire de Math\'ematiques, Universit\'e Paris-Sud 11, B\^at. 425, 91405 Orsay, France.} \email{jacques.smulevici@math.u-psud.fr}
\subjclass[2000]{Primary. 35Q76, 	83C05, 83C20}

\date{August 2014. Revised July 2015.}
\keywords{Einstein equations, T2--symmetry, future expanding spacetime, late-time asymptotics, geodesic completeness}

\begin{abstract} 
We investigate the late-time asymptotics of future expanding, polarized vacuum Einstein spacetimes 
with $T^2$-symmetry on $T^3$, which, by definition, admit two spacelike Killing fields. Our main result is the existence of a stable asymptotic regime within this class, that is, we provide here a full description of the late-time asymptotics of the solutions to the Einstein equations when the initial data set is close to the asymptotic regime. Our proof is based on several energy functionals with lower order corrections (as is standard for such problems) and the derivation of a simplified model that we exhibit here. Roughly speaking, the Einstein equations in the symmetry class under consideration consist of a system of wave equations coupled to constraint equations plus a system of ordinary differential equations. The unknowns involved in the system of ordinary equations are blowing up in the future timelike directions. One of our main contributions is the derivation of novel effective equations for suitably renormalized unknowns. Interestingly, this renormalization is not performed with respect to a fixed background, but does involve the energy of the coupled system of wave equations. In addition, we construct an open set of initial data that are arbitrarily close to the expected asymptotic behavior. 
We emphasize that, in comparison, the class of Gowdy spacetimes exhibits a very different dynamical behavior to the one we uncover in the present work for general polarized $T^2$--symmetric spacetimes. 
Furthermore, all the conclusions of this paper are valid within the framework of weakly $T^2$-symmetric spacetimes previously introduced by the authors.
\end{abstract}

\maketitle


  
\section{Introduction}

This is the third of a series of papers \cite{LeFlochSmulevici1, LeFlochSmulevici2} devoted to the study of weakly regular, $T^2$--symmetric vacuum spacetimes. There has been extensive work on the mathematical analysis of $T^2$--symmetric spacetimes with high regularity and we refer for instance to the introduction of  \cite{Smulevici2} for related literature. Our motivation in studying these spacetimes is two-fold. First of all, given the high degree of symmetry, one can study these solutions under much weaker regularity than in the general case. In \cite{LeFlochSmulevici1}, we introduced the notion of weakly regular, $T^2$--symmetric, vacuum spacetime and we established a future expanding, global existence theory in the so-called areal coordinates ---generalizing a previous result in the smooth setup \cite{BergerChruscielIsenbergMoncrief}. Our notion of weakly regular spacetimes extended a notion first proposed by Christodoulou \cite{Christodoulou}  (see also \cite{LeFlochSormani}) for radially symmetric spacetimes and later by LeFloch and co-authors  \cite{LeFloch,LeFlochMardare,LeFlochRendall,LeFlochStewart,LeFlochStewart2} for Gowdy symmetric spacetimes. See also the more recent developments in \cite{GL1,GL2}.

Our second motivation comes from the fact that,
 apart from special cases (see, for instance, 
\cite{ChruscielIsenbergMoncrief,Ringstrom, Ringstrom2}), a complete description of the late-time asymptotics of $T^2$ symmetric spacetimes has not been given yet {\sl even for smooth initial data sets.} In fact, the techniques available until now  provide the existence of future developments, but are not sufficient to prove that these spacetimes are future geodesically complete or not.

Recall that a $T^2$-symmetric, vacuum spacetime is a solution to the vacuum Einstein equations $Ric(g)=0$ arising from an initial data set which is assumed to be invariant under an action of the Lie group $T^2$. We are concerned here with the study of $T^2$-symmetric spacetime arising from initial data given on $T^3$. For such spacetimes, it is known \cite{Chrusciel} that, unless the spacetime is flat (and therefore the solution is trivial) the area of the orbits of symmetry, say $R$, admits a timelike gradient and, therefore, can be used as time coordinate and leads one to define the so-called areal gauge. By convention, we can choose the time direction so that $R$ increases toward the future. In the present paper, we restrict attention to \emph{polarized} $T^2$-symmetric spacetimes, which are $T^2$-symmetric spacetimes for which the Killing fields generating the $T^2$ symmetry can be chosen to be mutually orthogonal. 

Our main result is a complete description of the future time-asymptotics of polarized, $T^2$-symmetric, vacuum spacetimes, under the assumption that one starts sufficiently close to the expected asymptotic regime. As a consequence, it follows that these spacetimes are future geodesically complete. We refer to Theorem \ref{th:mt} and Theorem \ref{th:gc} for precise statements. These results are new even for smooth initial data, but we also emphasize that all of our estimates are valid within the framework of weakly regular $T^2$-symmetric spacetimes introduced in \cite{LeFlochSmulevici1}.

Prior to the present work, two important subclasses of $T^2$-symmetric solutions were studied in the literature. First of all, when the initial data set is invariant not only by an action of $T^2$ on $T^3$ 
but by the action of $T^3$ on itself, then the spacetime is homogeneous, i.e.~admits three independent spatial Killing fields. The Einstein equations then reduce to a set of ordinary differential equations. 
Second, another subclass of solutions is the class of Gowdy spacetimes, which, by definition, are $T^2$-symmetric solutions for which the family of $2$-planes orthogonal to the orbits of symmetry is integrable. One of the main differences between the Gowdy solutions and the general $T^2$-symmetric solutions is that the equations in areal gauge are semi-linear in the Gowdy case, while they are quasi-linear in general.  The future time-asymptotics of Gowdy spacetimes were derived by Ringstr\"om \cite{Ringstrom} (see also \cite{ChruscielIsenbergMoncrief} for polarized Gowdy spacetimes). 

The following  question thus arises. Are the asymptotics of homogeneous $T^2$-symmetric or Gowdy spacetimes {\sl stable within} the whole set of $T^2$-symmetric solutions? For homogeneous solutions, it turns out that there are not even stable within the class of Gowdy spacetimes \cite{Ringstrom}. As far as Gowdy spacetimes are concerned, the asymptotics derived in the present work
show that they are not stable within the set of $T^2$-symmetric solutions. For instance, according to Theorem \ref{th:mt} in Section \ref{se:ar}, the norm of the gradient of $R$ behaves like $R^{-2}$, while it decays exponentially in the Gowdy case. 
Of course, one question which remains open is whether the future asymptotic behavior that we uncover here is stable, first within the whole class of $T^2$-symmetric solutions (i.e~for non-polarized solutions) and, then, within the class of solutions arising from arbitrary initial data defined on $T^3$. We observe that many of the estimates we prove below can be generalized to the non-polarized case. 

Independently of this work, Ringstr\"om \cite{Ringstrom3} has recently obtained interesting and complementary results on $T^2$-symmetric spacetimes. His main results can be summarized as follows. For any $T^2$-symmetric spacetime which is non-flat and non-Gowdy, there is a certain geometric quantity\footnote{In the notation of this paper, it coincides with the quantity $\Pscr$ introduced in \eqref{defPcal} in Section \ref{se:hgs}.} which, if bounded as $R \rightarrow +\infty$,  implies that the solution is homogeneous. This result does not give sharp asymptotics on the solutions, but it is a large data result and therefore, it is so far the strongest result available for $T^2$-symmetric spacetimes with arbitrary data. It implies, in particular, that the asymptotics of non-Gowdy, non-homogeneous solutions are quite different from the asymptotics of homogeneous or Gowdy solutions. A second set of results proved in \cite{Ringstrom3} concerns polarized $T^2$-symmetric under a smallness assumption (which is slightly different from the initial data assumption that we make here). A partial set of asymptotics is then obtained therein, while in the present work, we derive a full set of late-time asymptotics; it is interesting to point out that the methods of proof appear to be quite different.

The rest of this paper is organized as follows. In the following section, we introduce standard material on $T^2$-symmetric and polarized solutions, which we will use throughout. In particular, we recall the global existence of areal foliation for weakly regular initial data established in \cite{LeFlochSmulevici1}. Apart from this result, this paper is essentially self-contained. We conclude the preliminary section by presenting the general strategy that we will use in order to derive the asymptotics. In Section \ref{se:evmv}, we derive some formulas for the evolution of certain mean values and we also provide some estimates about the commutator associated with the time derivative operator and the spatial average operator. Section \ref{se:mef} is devoted to the analysis of the corrected energy. In Section \ref{se:dsrq}, we introduce several renormalized unknowns, derive a system of evolution equations for them, and provide estimates on various error terms arising in the analysis. In the following section, we introduce and close a small bootstrap argument, linking all the previous estimates together. In Section \ref{se:ar} and \ref{se:gc}, we present and give the proofs of the main results of this paper, concerning the full set of asymptotics and the geodesic completeness of these spacetimes, respectively. Finally, in a final section, we construct an open set of initial data satisfying the assumptions of Theorem \ref{th:mt}. 
 

\section{Preliminaries on $T^2$-symmetric polarized solutions} 

\subsection{Einstein equations in areal coordinates}

Let $(\Mcal,g)$ be a weakly regular $T^2$-symmetric spacetime, understood in the sense introduced by the authors in \cite{LeFlochSmulevici1}. From the existence theory therein, we know that, if $R: \Mcal \to \RR$ denotes the area of the orbits of the symmetry group,
then its gradient vector field $\nabla R$ 
is timelike (and future oriented thanks to the standard normalization adopted in \cite{LeFlochSmulevici1}) 
and, consequently, the area can be used as a time coordinate. In 
these {\bf areal coordinates,}  the variable $R$ exhausts the interval $[R_0, +\infty)$,
where $R_0 > 0$ is the (assumed) constant value of the area on the initial slice, 
and the metric takes the form  
\be 
\aligned  
g = \, & e^{2(\eta-U)} \, \big( - dR^2 + a^{-2} \, d\theta^2 \big) 
 +e^{2U} \big(
dx + A \, dy + (G + AH) \, d\theta \big)^2
+ e^{-2U} R^2 \big( dy + H \, d\theta \big)^2.  
\endaligned 
\ee
Here, the independent variables $x,y$, and $\theta$ belong to $S^1$ (the one-dimensional torus or circle) and the metric coefficients $U,A,\eta, a, \displaystyle G$, and $H$ are functions of $(R,\theta)$, only. We will, for convenience in the presentation, identify $S^1$ with the interval $[0, 2\pi]$ and functions defined on $S^1$ with $2\pi$ periodic functions. The vector fields $\partial_x$ and $\partial_y$ are Killing fields for the above metric and so are any linear combinations of $\partial_x$ and $\partial_y$. 

We are interested here in \emph{polarized} $T^2$-symmetric spacetimes, defined as follows. 

\begin{definition}
A $T^2$-symmetric spacetime is said to be polarized if one can choose linear combinations $X,Y$ of the vector fields $\partial_x$, $\partial_y$ generating the $T^2$ symmetry such that $g(X,Y)=0$.
\end{definition}
For a polarized spacetime, it follows that the metric can be rewritten (possibly after a change of the coordinates $x,y$) as 
\be
\label{300} 
\aligned  
g = \, & e^{2(\eta-U)} \, \big( - dR^2 + a^{-2} \, d\theta^2 \big) 
 +e^{2U} \big(
dx + G \, d\theta \big)^2
+ e^{-2U} R^2 \big( dy + H \, d\theta \big)^2.
\endaligned 
\ee

Now, the Einstein equations for $T^2$-symmetric spacetimes written in areal coordinates have been derived in \cite{BergerChruscielIsenbergMoncrief} for smooth solutions (see also \cite{Chrusciel} for the existence of areal time). In \cite{LeFlochSmulevici1}, we introduced the {\bf weak version of the Einstein equations} for weakly regular $T^2$-symmetric spacetimes and we proved that, using areal coordinates, 
we could still reduce the Einstein equations to those obtained in \cite{BergerChruscielIsenbergMoncrief}. In the polarized case, we are thus left with the following system of partial differential equations:  
\begin{enumerate}

\item Three evolution equations for the metric coefficients $U, \eta, a$:
\begin{align}
\label{weakform1}
& (R \, a^{-1} U_R)_R - (R \, a \, U_\theta)_\theta 
= 0, 
\\
\label{weakform3} 
&( a^{-1} \eta_R)_R - ( a \, \eta_\theta)_\theta 
= \Omega^\eta 
- {1 \over R^{3/2}} \big( R^{3/2} \big( a^{-1} \big)_R\big)_R\,\, ,
\\
\label{eq:lnalphar}
&( 2 \ln a )_R = - {K^2 \over R^3} \, e^{2 \eta}, 
\end{align} 
where $K$ is a real constant and 
$\Omega^\eta :=  - a^{-1} U_R^2 + a \, U_\theta^2$.

\item Two constraint equations for the metric coefficient $\eta$: 
\begin{eqnarray}
\label{weakconstraintsr1}
\eta_R + {K^2 \over 4 R^3} \, e^{2\eta}&=& a \, RE,\\
\eta_\theta  &=& R \, F, \label{weakconstraintsr2}
\end{eqnarray}
where 
$E:= a^{-1} \, U_R^2 + a \, U_\theta^2$ and 
$F := 2 \, U_R U_\theta$. 

\item Two equations for the twists: 
\be
\label{GHequa} 
G_R = 0, 
\qquad H_R = {K \over R^3}  \, a^{-1} e^{2 \eta}.
\ee
\end{enumerate} 
Here, 
$K$ is the \emph{twist constant} and $K=0$ corresponds geometrically to the integrability of the family of $2$-planes orthogonal to $\partial_x$ and $\partial_y$. The special solutions with $K=0$ are called Gowdy spacetimes (with $T^3$ topology). Since the dynamics of Gowdy spacetimes are well-known \cite{Ringstrom}, we focus here exclusively on the case $K\neq0$.

Note that the metric functions $G$ and $H$ do not appear in the equations apart from \eqref{GHequa}. These latter equations can simply be integrated in $R$, once enough information on their right-hand sides is obtained. They will therefore be ignored in most parts of this paper. Note also that \eqref{weakform3} is actually a redundant equation, i.e.~can be deduced from the other equations\footnote{More precisely, \eqref{weakform3} can be obtained by multiplying \eqref{weakconstraintsr1} and \eqref{weakconstraintsr2} by $a^{-1}$ and $a$ respectively, differentiating the resulting equations in $R$ and $\theta$ and taking their differences before replacing second derivatives of $U$ and first derivatives of $a$ using the evolution equations.}.

Finally, observe that the identity 
\begin{equation}\label{id:ea}
\left( \frac{ e^{2\eta}}{a} \right)_R = 2 R E e^{2\eta},
\end{equation}
will be useful later in this paper; it can be easily derived from the Einstein equations  \eqref{eq:lnalphar} and \eqref{weakconstraintsr1}.


\subsection{Global existence in areal coordinates}

In \cite{LeFlochSmulevici1}, we proved local and global existence results for general $T^2$-symmetric spacetimes in areal coordinates. In the specific case of polarized $T^2$-symmetric spacetimes, these results imply the following conclusion. 

\begin{theorem}[Global existence theory in areal coordinates]
Fix any constants $K,R_0 >0$. consider any initial data 
$(U_0,U_1) \in H^1(S^1) \times L^2(S^1)$, 
$a_0 \in W^{2,1}(S^1)$, and $\eta_0 \in W^{1,1}(S^1)$ such that $a_0>0$.
 Suppose moreover that the constraint equation \eqref{weakconstraintsr2} is satisfied initially i.e. 
\be
\partial_\theta(\eta_0)  = 2 R_0 \,  U_1 \partial_\theta(U_0).
\ee
Let $\mathcal{C}$ be the class of functions $(U,\eta,a)$ such that 
$$
\aligned
&
U\in C^1([R_0,+\infty), L^2(S^1)) \cap C^0([R_0,+\infty); H^1(S^1)), 
\\
& \eta \in C^0([R_0,+\infty); W^{1,1}(S^1)),
a \in C^0([R_0,+\infty); W^{2,1}(S^1)).
\endaligned
$$
Then there exists a unique solution $(U,\eta,a) \in \mathcal{C}$ of the Einstein equations \eqref{weakform1}-\eqref{weakconstraintsr2} which assumes the given initial data at $R=R_0$, in the sense
\begin{eqnarray*}
U(R_0)&=& U_0, \quad U_R(R_0)=U_1, 
\\
\eta(R_0)&=&\eta_0, \quad a(R_0)=a_0.
\end{eqnarray*}
Moreover, on any compact time interval, the solution can be uniformly approximated by smooth solutions in the norm associated with $\mathcal{C}$.
\end{theorem}

Since all of our estimates here will be compatible with the density property stated 
at the end of the above theorem, it is sufficient to perform our analysis by assuming our initial data to be smooth. 


\subsection{Energy functionals}

Important control on the metric coefficients, mostly on their first-order derivatives,
is obtained by analyzing the energy functionals 
\be
\Escr(R)  := \int_{S^1} E(R,\theta) \, d\theta, \qquad E=  a^{-1} \, U_R^2 + a \, U_\theta^2,
\ee
and 
\be
\aligned 
\Escr_K(R)  :=& \int_{S^1} E_K(R,\theta) \, d\theta,
\qquad
E_K  := \, E + {K^2  \over 4 R^4} \,  a^{-1}  e^{2 \eta}. 
\endaligned 
\ee
Using the Einstein equations \eqref{weakform1}--\eqref{weakconstraintsr2}, it follows that both functionals 
are {\sl non-increasing} in time, with 
\be
\label{decay1}
\aligned
&{d \over dR} \Escr(R) 
= - {K^2 \over 2 R^3} \int_{S^1} E \, e^{2 \eta} \, d\theta 
- {2 \over R} \int_{S^1} 
a^{-1}  \, (U_R)^2  d\theta, 
\\
& {d \over dR} \Escr_K(R) 
= - {K^2 \over  R^5} \int_{S^1} \, a^{-1}e^{2 \eta} \, d\theta 
- {2 \over R} \int_{S^1} 
a^{-1}  \, (U_R)^2   \, d\theta.
\endaligned
\ee 
 As a direct consequence, we have the following result. 

\begin{lemma}[Uniform energy bounds for $T^2$--symmetric spacetimes] 
\label{energ}
The following uniform bounds hold
\be
\label{E400} 
\aligned 
&\sup_{R \in [R_0,+\infty)} \Escr(R) \leq \Escr(R_0), 
\qquad 
\qquad 
\sup_{R \in [R_0,+\infty)} \Escr_K(R) \leq \Escr_K(R_0), 
\endaligned 
\ee
as well as the spacetime bounds
\be
\label{Espacetime}
\int_{R_0}^{+\infty} \int_{S^1} \Big( 
a^{-1}  c_0^U   \, (U_R)^2 
+ 
a \, c_1^U  \, (U_\theta)^2 
 \Big) \, dRd\theta \leq \Escr(R_0), 
\ee
\be
\label{Espacetime2}
\int_{R_0}^{+\infty} {K^2 \over R^5} \int_{S^1} e^{2 \eta} \, a^{-1} \, dR d\theta \leq \Escr_K(R_0),
\ee 
with 
$$
\aligned
& c_0^U := {2 \over R}   + {K^2 \over 2R^3} \, e^{2 \eta}, 
\qquad 
c_1^U := {K^2 \over 2R^3} \, e^{2 \eta}.
\endaligned
$$ 
\end{lemma}


\subsection{Heuristics and general strategy} \label{se:hgs}

To understand the asymptotic behavior of the solutions to wave equations such as \eqref{weakform1}, it is important to note that while, for the flat wave operator in $1+1$ dimensions, there is no decay of solutions, the $R$-weights present in \eqref{weakform1} reflect some expansion of our spacetime and that in general, waves decay on expanding spacetimes.

The general strategy to capture this decay is to first observe that the 
global energy dissipation bound \eqref{Espacetime} associated with the energy functional $\Escr(R)$ gives an integrated energy decay estimate but with weaker weights for $U_\theta$ than for $U_R$ (see the missing $2/R$ in $c_1^U$ compared to $c_0^U$). 
To match the weights between $U_R$ and $U_\theta$, we will work instead
with 
the {\bf modified energy functional} 
\be
\label{ModE} 
\Escrhat(R) := \Escr(R) +  \Gscr^U(R) 
\ee
with 
$$
\aligned
\Gscr^U &:= {1 \over R} \int_{S^1} \big( U - \la U \ra \big) \, U_R \, a^{-1} \, d\theta,  
\endaligned
$$
in which the average $\la f \ra$ of a function $f=f(\theta)$ is not defined with respect to the flat measure $d\theta$ but with respect to a weighted measure $a^{-1} d\theta$, i.e.~
\begin{eqnarray} \label{def:avera}
\la f \ra := {\int_{S^1} f \, a^{-1} d\theta \over \int_{S^1} a^{-1} d\theta}. 
\end{eqnarray}
Our strategy is then to ``trade'' a time-derivative for a space-derivative. This method of proof was previously used in \cite{Ringstrom, ChoquetBruhatMoncrief,ChoquetBruhat}.

The following notation will be useful. 
We introduce the length $\Pscr$ of the circle $S^1$ with respect to the measure $a^{-1}d\theta$, that is, 
\be
\label{defPcal} 
\Pscr(R):=\int^{2\pi}_0 a^{-1} d\theta,
\ee
which we refer to as the {\bf perimeter.} The geometric interpretation of this quantity is that the principal symbol of the wave operator appearing in the wave equation \eqref{weakform1} for $U$ is that of the $2$ dimensional metric
$$
ds^2=-dR^2+a^{-2}d\theta^2.
$$
Thus, $\Pscr(R)$ is the volume of the constant-$R$ slices for this metric.

Naively, one may expect the following behavior as $R \to +\infty$. 
In view of the energy identity \eqref{decay1} satisfied by $\Escr$ and focusing on the second integral term, 
one may expect that 
$$
{d \over dR} \Escr \leq - {2 \over R} \Escr \qquad \text{ (modulo higher order terms),} 
$$
so that $\Escr$ should decay like $\frac{1}{R^2}$. 
This behavior is indeed correct for spatially homogeneous spacetimes, as can be checked directly.  
However, for non-spatially homogeneous solutions, a space-derivative must be recovered from a time-derivative, using the corrected energy $\Escrhat$ defined in \eqref{ModE}, as we already explained above.
This would lead to a rate of decay determined by 
$$
{d \over dR} \Escrhat(R) \leq - {1 \over R} \Escrhat(R) \qquad \text{ (modulo higher order terms),} 
$$
so that $\Escrhat$ should decay like $\frac{1}{R}$. If one can then check that the correction term in $\Escrhat$ is of order $o(1/R)$, it should follow that $\Escr(R)$ is of order $1/R$. This is indeed the rate of decay established by Ringstr\"om~\cite{Ringstrom} for (sufficiently regular) Gowdy spacetimes. 

For the more general class of spacetimes under consideration in the present paper, and 
due to the variation of the metric coefficients $a$ and $\eta$, the
behavior $\Escr \sim 1/R$ is {\sl not} consistent with the field equations, as we now check formally.
At this stage of the discussion, 
we are working under the (later invalidated, below) assumption that the first term in \eqref{decay1} is negligible, 
say specifically 
\be
\label{407}
{||e^{2 \eta}||_{L^\infty(S_1)} \over 2R^3} \lesssim {1 \over R^{1 + \eps}}, \qquad \eps >0. 
\ee
From \eqref{eq:lnalphar} we would deduce 
$$
\aligned
(\ln a)_R & = - K^2 \, {e^{2 \eta} \over 2R^3} \in L^1_R, 
\endaligned
$$
hence the coefficient $a$ would then admit a finite limit as $R \to +\infty$. Next, in view of  
$$
\eta_R  = - \frac{K^2}{2} \, {e^{2 \eta} \over 2R^3}  + a \, RE,
$$
in which $\int_{S^1}R E d\theta $ is bounded thanks to our energy assumption, 
it would then follow that $\int_{S^1}\eta_R$ behaves like $1$ and thus $\int_{S_1}\eta \sim R$ (modulo a multiplicative constant). 
In turn, this invalidates our original assumption \eqref{407}. 

This means that the first term in \eqref{decay1} should not be neglected and that it contributes significantly to the energy decay. 
We will prove that, modulo an error term due to the spatial variation of $\eta$, this term can be rewritten as $-\frac{\Pscr_R}{\Pscr} \Escr$, where $\Pscr$ is the perimeter defined by \eqref{defPcal}.

Taking this into account, it follows, assuming that all the error terms can be controlled, that the {\bf rescaled energy}  
\be
\Fcal := \Pscr \, \Escrhat
\ee
should decay like $\frac{1}{R}$ and, in other words, the energy $\Escrhat$ should decay like $\frac{1}{\Pscr R}$. This brings more decay into our analysis, provided the perimeter $\Pscr$ is {\sl growing as $R \to +\infty$} ---as 
we will actually show later. Indeed, we will establish that the perimeter and metric coefficients have the following asymptotic behavior (possibly up to multiplicative constants):  
\be
\label{444}
\aligned
\Pscr(R) & \sim R^{1/2}, \qquad &&\Pscr_R(R) \sim R^{-1/2}, 
\\
e^{2\eta} & \sim R^2, \qquad &&a         \sim R^{-1/2}.
\endaligned
\ee
For the energy, we will therefore have $\Escr \sim R^{-3/2}$. Surprisingly, all the multiplicative constants in the above asymptotic behavior are linked to each other. For instance, we will show that $R^2 \Pscr^{-1} \Escr \rightarrow 5/4$ as $R \rightarrow +\infty$. One of the main difficulties lies in fact in trying to understand these relations. Thus, our work really consists of three main ingredients: 
\begin{enumerate}
\item a version of the corrected energy functionals adapted to polarized $T^2$-symmetric spacetimes (Sections \ref{se:evmv} and \ref{se:mef}),
\item a derivation and analysis of a dynamical system to understand the interplay between $\Pscr$ and the energy functionals (Section \ref{se:dsrq}),
\item and estimates on all the error terms involved in the above two steps and the interplay between all the previous estimates. Since all the estimates involved in the above estimates depend on each other, we use a small bootstrap argument to obtain closure (Section \ref{se:sdp}).
\end{enumerate}
Once these elements have been obtained, deriving the asymptotics of the solutions consists mostly in revisiting the previous estimates in the proper order (see Section \ref{se:ar}, below). Finally, 
in Section \ref{se:gc},
we prove the geodesic completeness by using the approach already developed in \cite{LeFlochSmulevici2}.


\section{Evolution of the mean values} \label{se:evmv}
\label{section-ME} 

\subsection{The length variable} 

In addition to the perimeter $\Pscr(R)$ introduced in \eqref{defPcal}, the metric coefficient $a$ also determines a {\bf length function}  
\be
\vartheta(\theta,R) :=\int^\theta_0 a^{-1} d\theta, \qquad \theta \in S^1, 
\ee
and its inverse  $\Theta=\Theta(\vartheta,R)$ (for each fixed $R$). In other words, we set 
$\Theta(\vartheta(\theta,R),R) = \theta$ for all $\theta \in S^1$, 
so that
\be
\Theta(\vartheta,R) = \int^\vartheta_0 a (\Theta(\vartheta',R),R) \, d\vartheta',
\qquad 
\Theta( \Pscr(R),R) =2\pi.
\ee

Using the change of variable determined by the length function, we can parameterize any function $f= f(R,\theta)$ 
into $\tildef=\tildef(R,\vartheta)$, defined by
\be
\label{tildenotat}
\tildef(R,\vartheta) := f\left(R,\Theta(\vartheta,R)\right). 
\ee
This is nothing but a change of coordinates from $(R,\theta)$ to $(R,\vartheta)$, but we insist on keeping the "tilde notation" in order to avoid confusion (when taking averages and $R$ derivatives).

The average of any $L^1(S^1)$ function $f$ is now naturally computed with respect to the measure $d\vartheta$, that is,  
\be
\aligned
\la \tildef(R) \ra
:= &\frac{1}{\Pscr(R)} \int^{\Pscr(R)}_0 \tildef(R) \, d\vartheta.
\\
 = &\frac{1}{\Pscr(R)}\int^{2\pi}_0 f(R) \, a(R)^{-1} d\theta= \la f(R) \ra,
\endaligned
\ee
which, as stated, obviously coincides with $\la f(R) \ra$ as defined by \eqref{def:avera}.
Note that the periodicity property is preserved in the new variable, that is, 
$$
\tildef (R,\vartheta+\Pscr(R))=\tildef(R,\vartheta).
$$
for all relevant values of $R$ and $\vartheta$. 

Using the above notation, we can for instance rewrite the correction $\Gscr^U$ introduced in \eqref{ModE} in the form 
\be
\label{eq:guvartheta}
\aligned
\Gscr^U(R) :=& \frac{1}{R}\int^{\Pscr(R)}_0 \left( \tildeU(R) - \la\tildeU(R)\ra\right) \tildeU_R (R) \,  d\vartheta.
\endaligned
\ee 
This form has some advantages when differentiating with respect to $R$, since it directly involves the perimeter and its derivative, which have a geometric meaning. 


\subsection{Derivatives of the mean values}  

We will be taking time-derivatives of the above quantities but since the time-derivative operator and the spatial averaging operator {\sl do not commute}, an analysis of the corresponding ``commutator'' will be required.
The following properties will be used throughout the rest of this article.

\begin{lemma}[General identities for the mean values] 
\label{meanv} 
For any (sufficiently regular) function $f=f(R,\theta)$,  
one has 
$$
\aligned
\frac{d}{dR}\aver{f} & =\aver{f_R}+\frac{K^2}{2R^3}\aver{fe^{2\eta}}
-\frac{\Pscr_R }{\Pscr}\aver{f},
\\
\frac{d}{dR} \Big( \Pscr \aver{f} \Big) & = \Pscr \aver{f_R}
+ \Pscr \frac{K^2}{2R^3}\aver{fe^{2\eta}}, 
\endaligned
$$
in which $\tildef$ is defined by \eqref{tildenotat}.
\end{lemma}

\begin{proof} From the definition 
$$
\aver{f}=\frac{1}{\Pscr}\int^{\Pscr}_0\tildef d\vartheta
=\frac{1}{\Pscr }\int^{2\pi}_0 f(R,\theta) a^{-1}d\theta,
$$
we deduce that 
$$
\aligned
\frac{d}{dR} \aver{f} 
=& \aver{f_R}+\frac{1}{\Pscr }\int^{2\pi}_0f (a^{-1})_Rd\theta -\frac{\Pscr_R }{\Pscr}\aver{f}
\\
=&  \aver{f_R}+ \frac{1}{\Pscr }\int^{2\pi}_0 f \frac{K^2 e^{2\eta}}{2R^3} a^{-1} d\theta
-\frac{\Pscr_R }{\Pscr}\aver{f} 
\\
=&  \aver{f_R}+\frac{K^2}{2R^3}\aver{fe^{2\eta}}
-\frac{\Pscr_R }{\Pscr}\aver{f},
\endaligned
$$
which leads us to the two identities stated in the lemma. 
\end{proof}

The above lemma allows us to derive the following estimate. 

\begin{lemma}[Commutator estimate] 
\label{601}
The commutator associated with the time-differentiation and averaging operators satisfies 
for all functions $f$ 
$$
\left| \frac{d}{dR}\aver{f}- \aver{f_R} \right| \le 
\frac{ \pi K^2}{R^3} \, \aver{|f|} \big\| \big( e^{2\eta} \big)_\theta \big\|_{L^1(S^1)}.
$$
\end{lemma}

\begin{proof} From the above lemma, the expression of $\Pscr_R$, and the evolution equation satisfied by $a$, 
we deduce 
$$
\aligned
&
\Big| \frac{d}{dR}\aver{f} -  \aver{f_R}\Big|
\\
& \leq
 \frac{K^2}{2R^3\Pscr^2} \int^{2\pi}_0 |f|\, a^{-1}(R,\theta) 
\left|\, e^{2\eta}(R,\theta) \Pscr-\int^{2\pi}_0 e^{2\eta}a^{-1}(R,\theta')\, d\theta' \right| d\theta
\\
& \leq 
 \frac{ \pi K^2}{R^3\Pscr} \aver{|f|} \sup_{\theta \in S^1} 
\left|\, e^{2\eta}(R,\theta) \Pscr-\int^{2\pi}_0 e^{2\eta}(R,\theta') a^{-1}(R,\theta') d\theta'\right|
\endaligned
$$
with 
$$
\aligned
\sup_{\theta \in S^1} \left|\, e^{2\eta}(R,\theta) \Pscr
-\int^{2\pi}_0 e^{2\eta}(R,\theta')\, a^{-1}(R,\theta')\,d\theta' \right|
& \le \Pscr\left( \sup_{S^1} e^{2 \eta} - \min_{S^1} e^{2\eta} \right)
\\
& \le \Pscr \, \big\| \big( e^{2\eta} \big)_\theta \big\|_{L^1(S^1)}.
\endaligned
$$
\end{proof}

The following conserved quantity will also be useful in our analysis. It follows simply after a global integration in space of the wave equation \eqref{weakform1} and an integration in $R$ on $[R_1,R]$.

\begin{lemma}
\label{eq:mvar}
For all $R\ge R_1$, the following conservation law holds: 
$$
\aligned
 R \Pscr \, \aver{U_R}
= R_1 \Pscr(R_1) \, \aver{U_R}(R_1).
\endaligned
$$
\end{lemma}


\section{Evolution of the modified energy functional}
\label{se:mef} 

\subsection{Evolution of correction terms}

Using Lemma~\ref{meanv}, we can compute the time derivative of the corrector $\Gscr^U$  in \eqref{eq:guvartheta}, indeed: 
\begin{eqnarray}
\frac{d}{dR}\Gscr^U &= -\frac{1}{R}\Gscr^U&+\frac{1}{R}\left(\int^{2\pi}_0\left( U - \aver{U} \right) U_R a^{-1}d\theta \right)_R \nonumber \\
&=-\frac{1}{R}\Gscr^U&+\frac{1}{R}\int^{2\pi}_0 U_{R}^2 a^{-1}d\theta \nonumber\\
&&+\frac{1}{R}\int^{\Pscr}_0\left(-\aver{U_R}-\frac{K^2}{2R^3}\aver{Ue^{2\eta}}+\frac{\Pscr_R}{\Pscr}\aver{U}\right)\tildeU_R d\vartheta \nonumber\\
&&+\frac{1}{R} \int^{2\pi}_0\left( U - \aver{U} \right)  ( U_R a^{-1} )_R d\theta,  \nonumber
\end{eqnarray}
so that, by using the field equation \eqref{weakform1} satisfied by $U$,  
\begin{eqnarray} 
\frac{d}{dR}\Gscr^U &=&-\frac{1}{R}\Gscr^U+\frac{1}{R}\int_0^{\Pscr} \widetilde{U_R}^2 d\vartheta \nonumber\\
&&-\frac{\Pscr}{R}\left( \aver{U_R} \right)^2 -\frac{K^2}{2R^4}\Pscr\aver{U_R}\aver{Ue^{2 \eta }}+\frac{\Pscr_R}{R}\aver{U}\aver{U_R} \nonumber\\
&&+\frac{1}{R}\int^{2\pi}_0\left( U - \aver{U} \right) \left( -\frac{U_R a^{-1}}{R}+\left( a U_\theta \right)_\theta\right) d\theta. \nonumber
\end{eqnarray}
Integrating by parts the last term, we obtain
\begin{eqnarray*} 
&&-\frac{2}{R}\Gscr^U+\frac{1}{R}\int_0^{\Pscr} \widetilde{U_R}^2 d\vartheta-\frac{1}{R}\int_0^{\Pscr}{\tilde{U}_{\vartheta}^2}d\vartheta\\
&&-\frac{\Pscr}{R}\left( \aver{U_R} \right)^2-\frac{K^2}{2R^4}\Pscr\aver{U_R}\aver{Ue^{2 \eta }} +\frac{\Pscr_R}{R}\aver{U}\aver{U_R}. 
\end{eqnarray*}
After re-organizing some of the terms, this leads us to  
\be
 \label{eq:dgammau}
\aligned 
\frac{d}{dR}\Gscr^U  
&=-\frac{1}{R}\int_0^{\Pscr}{\tilde{U}_{\vartheta}^2}d\vartheta +\frac{1}{R}\int_0^{\Pscr} \widetilde{U_R}^2 d\vartheta  \\
& \quad -\frac{1}{R}\Gscr^U-\frac{\Pscr_R}{\Pscr}\Gscr^U 
+\Omega_{\Gscr^U}, 
\endaligned
\ee
with
\be
\label{eq:errgammau} 
\aligned
\Omega_{\Gscr^U}
&=\frac{\Pscr_R}{\Pscr}\Gscr^U-\frac{\Pscr}{R}\left( \aver{U_R} \right)^2 -\frac{1}{R}\Gscr^U 
\\
& \quad-\frac{K^2}{2R^4}\Pscr\aver{U_R}\aver{Ue^{2\eta}}+\frac{\Pscr_R}{R}\aver{U} \aver{U_R}.
\endaligned
\ee 
The term $\Omega_{\Gscr^U}$ will be shown to be an ``error term'', while the remaining terms in the right-hand side of \eqref{eq:dgammau} will contribute to the derivation of a sharp energy decay estimate. In \eqref{eq:dgammau} and \eqref{eq:errgammau}, we have added and subtracted the term $\frac{\Pscr_R}{\Pscr}\Gscr^U$, as this will simplify some of our estimates.


\subsection{Evolution of the corrected energy}

Adding together the contributions of the energy and the correction $\Gscr^U$, we find 
\be
 \label{eq:cef}
\aligned
{d \over dR} \Big( \Escr + \Gscr^U \Big) 
&=-\frac{K^2}{2R^3} \int_{0}^{2\pi} E e^{2\eta}d\theta - \frac{2}{R} \int_{S^1} \left(a^{-1} U_R^2 \right)d\theta 
 \\
&\hbox{}+\frac{1}{R} \int_{0}^{2\pi} a^{-1} U_R^2 d\theta - \frac{1}{R} \int^{\Pscr}_0 \tilde{U}_{\vartheta}^2 d\vartheta   -\frac{\Pscr_R}{\Pscr}\Gscr^U- \frac{1}{R}\Gscr^U  
+\Omega_{\Gscr^U} 
\\
&=-\frac{\Pscr_R}{\Pscr}\left( \Escr+\Gscr^U\right) 
-\frac{1}{R}\left(\Escr+\Gscr^U\right) +\Omega_\Escr+\Omega_{\Gscr^U}, 
\endaligned
\ee
where the error terms are $\Omega_{\Gscr^U}$ defined by \eqref{eq:errgammau} and
$$
\Omega_\Escr=\frac{\Pscr_R}{\Pscr}\Escr-\frac{K^2}{2R^3} \int_{0}^{2\pi} E e^{2\eta}d\theta.
$$

\subsection{Estimate for the energy correction}

We will need the following $1$-dimensional Poincar\'e (or Wirtinger) inequality: for any $a > 0$, if $f$ is an $a$-periodic function in $H^1(0,a)$ and has $0$ mean value on this interval, then 
\be
\int_{[0,\,a]} f^2 \le \frac{a^2}{4 \pi^2} \int_{[0,\,a]} f'^2. 
\ee
This is easily checked by, for instance, using a Fourier decomposition of $f$. 
Using the above notation, we have the following lemma.

\begin{lemma}[Estimate of the $\Gscr^U$ correction of the energy] 
\label{lem:egu}
One has 
$$
|\Gscr^U(R)| \le \frac{ \Pscr(R)}{4\pi R} \Escr(R).
$$
\end{lemma}

\begin{proof}We apply the inequality $ab \le (a^2+b^2)/2$ 
to the integrand of $R \Gscr^U$, but we insert weights of
$\frac{\Pscr}{2 \pi}$ so as to obtain 
\begin{eqnarray*}
|R \Gscr^U | &\le& \frac{\Pscr}{4 \pi} \int_0^{\Pscr} \widetilde{U_R}^2 d\vartheta 
+\frac{2 \pi}{2\Pscr} \int_0^{\Pscr}  \left( \tildeU-<\tildeU> \right)^2 d\vartheta,\\
 &\le& 
\frac{\Pscr}{4 \pi} \int_0^{2\pi} U_R^2 a^{-1} d\theta
+\frac{\Pscr}{4 \pi} \int_0^{\Pscr} \tildeU_{\vartheta}^2 d\vartheta
=\frac{\Pscr}{4\pi}\Escr.
\end{eqnarray*} 
\end{proof}
 

\subsection{Estimates for the error terms} 

In this section, we estimate all the error arising in the corrected energy formula \eqref{eq:cef}.
 
\begin{lemma}[Estimate for the $|\Omega_\Escr|$ error term]
\label{lem:eee}
One has 
$$
|\Omega_\Escr | \le \Escr \frac{K^2}{2R^3} \int_0^{2\pi} 2R E e^{2\eta} = \Escr \frac{K^2}{2R^3} \int_0^{2\pi} \left( \frac{e^{2\eta}}{a} \right)_R. 
$$
\end{lemma}

\begin{proof} Recall that 
$$
\frac{\Pscr_R}{\Pscr}= \frac{K^2 }{2R^3} \int_0^{2\pi}e^{2\eta} a^{-1}d\theta \left( \int_0^{2\pi} a^{-1} d\theta\right)^{-1}, 
$$
so that  
$$
\aligned
& \left|-\frac{K^2}{2R^3} \int_{0}^{2\pi} E e^{2\eta}d\theta+\frac{\Pscr_R}{\Pscr}\Escr \right| 
\\
&\le \frac{K^2}{2R^3 \Pscr} \int_0^{2\pi} E(R,\theta) d\theta \int_0^{2\pi} a^{-1}(R,\theta') | e^{2\eta}(R,\theta')-e^{2\eta}(R,\theta)| d\theta' 
\\
&\le \Escr \frac{K^2}{2R^3} \int_0^{2\pi} | 2\eta_\theta| e^{2\eta} d\theta 
\\
&\le \Escr \frac{K^2}{2R^3} \int_0^{2\pi} 2R E e^{2\eta}d\theta = \Escr \frac{K^2}{2R^3} \int_0^{2\pi} \left( \frac{e^{2\eta}}{a} \right)_R d\theta,
\endaligned 
$$
where we have used the constraint equation for $\eta_{\theta}$ \eqref{weakconstraintsr2} and the identity \eqref{id:ea}.
\end{proof}

Next, we analyse the error term $\Omega_{\Gscr^U}$. It is convenient to split it into three components as follows: 
$\Omega_{\Gscr^U}=I_1+I_2+I_3$, where $I_1$, $I_2$, $I_3$ are defined as
\begin{eqnarray}
I_1&=&-\frac{1}{R}\Gscr^U, \nonumber \\
I_2&=&\frac{\Pscr_R}{\Pscr}\Gscr^U+\frac{\Pscr_R}{R}\aver{U} \aver{U_R} -\frac{K^2}{2R^4}\Pscr\aver{U_R}\aver{Ue^{2\eta}}, \nonumber \\
I_3&=&-\frac{\Pscr}{R}\left( \aver{U_R} \right)^2.
\nonumber
\end{eqnarray}

\begin{lemma}\label{lem:Iee}The following estimates hold
\begin{eqnarray}
|I_1 | &\le& \frac{ \Pscr(R)}{4\pi R^2} \Escr(R), \nonumber \\
|I_2| &\le& \frac{\Pscr_R}{R}\Escr\,,\nonumber \\
|I_3| &\le& \frac{\mathcal{A}}{R^3  \Pscr(R) }, \nonumber
\end{eqnarray}
where $\mathcal{A}$ is a non-negative constant determined by the initial data
$$\mathcal{A}=R_1^2 \Pscr(R_1)^2 \left(\aver{U_R} \right)^2(R_1).$$
\end{lemma} 

\begin{proof} The estimates on $I_1$ and $I_3$ follow immediately from Lemmas \ref{lem:egu} and \ref{eq:mvar}, respectively. We then estimate $I_2$ as follows. Note first that
\begin{eqnarray*}
I_2&=&\frac{\Pscr_R}{\Pscr}\Gscr^U+\frac{\Pscr_R}{R}\aver{U} \aver{U_R} -\frac{K^2}{2R^4}\Pscr\aver{U_R}\aver{Ue^{2\eta}},\\
&=&\hbox{}\frac{K^2}{2R^4 \Pscr}\int_0^{2\pi} U_R(R,\theta') a^{-1}(R,\theta') \left( \int_0^{2\pi}e^{2\eta}(R,\theta)a^{-1}(R,\theta) [U(R,\theta')-U(R,\theta)]d\theta \right) d\theta', 
\end{eqnarray*}
hence
\begin{eqnarray} 
|I_2|&\le& 
\frac{K^2}{2R^4 \Pscr}\int_0^{\Pscr}|\widetilde{U_R}|\,d\vartheta \int_0^{\Pscr} e^{2\tilde{\eta}}\, d\vartheta\int_0^{\Pscr} |\tildeU_\vartheta|\,d\vartheta, \nonumber \\
&\le& \frac{\Pscr_R}{R\Pscr} \left(\Escr^{1/2} \Pscr^{1/2} \right)^2 \nonumber
\le \frac{\Pscr_R}{R}\Escr.\nonumber
\end{eqnarray}
\end{proof}
 

\subsection{Combining the estimates for the corrected energy}

Collecting all the estimates for the error terms above, and noting that $I_3$ has a sign, we obtain the  estimate
$$
\aligned
\label{eq:pit}
\frac{d}{dR} \left( \Escr +\Gscr^U \right) + \left( \frac{1}{R}+\frac{\Pscr_R}{\Pscr} \right) \left( \Escr +\Gscr^U \right) \le \frac{\Pscr}{4\pi R^2} \Escr + \frac{\Pscr_R}{R}\Escr+\Escr \frac{K^2}{2R^3} \int_0^{2\pi}\left( \frac{e^{2\eta}}{a} \right)_R, 
\endaligned
$$
from which it follows that
$$
\aligned
R \Pscr (\Escr +\Gscr^U )(R) 
& \le R_0\Pscr(\Escr+\Gscr^U)(R_0)+\int_{R_0}^R \frac{\Pscr^2 \Escr}{4\pi R'} dR' +\int_{R_0}^R \Pscr_R \Pscr \Escr dR' \\
& \quad +\int_{R_0}^{R} \Pscr \Escr \frac{K^2}{2{R'}^2} \int_0^{2\pi} \left( \frac{e^{2\eta}}{a} \right)_R d\theta dR'.
\endaligned
$$
Similarly, we can obtain 

\begin{eqnarray} \label{dineq:glb}
\frac{d}{dR} \left( R \Pscr (\Escr +\Gscr^U ) \right) \ge -\frac{\Acal}{R^2} -\frac{\Pscr^2}{4\pi R} \Escr + \Pscr \Pscr_R \Escr+ \Pscr \Escr \frac{K^2}{2R^2} \int_0^{2\pi}\left( \frac{e^{2\eta}}{a} \right)_R, 
\end{eqnarray} 
leading to 
$$
\aligned
R \Pscr (\Escr +\Gscr^U )(R) 
& \ge R_0\Pscr(\Escr+\Gscr^U)(R_0)-\int_{R_0}^R \frac{\Pscr^2 \Escr}{4\pi R'} dR'
 \\
& -\int_{R_0}^R \Pscr_R \Pscr \Escr dR' -\int_{R_0}^{R} \Pscr \Escr \frac{K^2}{2{R'}^2} \int_0^{2\pi} \left( \frac{e^{2\eta}}{a} \right)_R d\theta dR'-\int_{R_0}^{R}\frac{\mathcal{A}}{R'^2}dR',
\endaligned
$$
where $\mathcal{A}$ is the constant in Lemma \ref{lem:Iee}.  


\section{A dynamical system for the renormalized unknowns}
\label{se:dsrq} 

\subsection{The dynamical system}

In the previous section, we have obtained differential inequalities for the quantity $\Pscr (\Escr +\Gscr^U )$, with error terms depending mostly on $\Escr$ and $\Pscr$. In this section, we will try to obtain effective equations in order to control the asymptotic behavior of $\Pscr$. For convenience, we introduce the notation
$$
\aligned
\Fcal&:=\Pscr \Escr\,\\
\Gcal&:=\Pscr (\Escr +\Gscr^U ).
\endaligned
$$
We have thus seen that $\Gcal$ satisfies ``good'' differential inequalities while it is ultimately $\Fcal$ that we want to control, as it is a manifestly coercive quantity (contrary to $\Gcal$). 
We will rely on the guess that the function $\Gcal$ decays like $1/R$ but we will not use yet the differential inequalities derived for $\Gcal$ in the previous section. In fact, $\Gcal$ will appear here only in the form $R \Gcal'/\Gcal$.

\subsubsection*{The system of odes: spatial integration and first error terms}
 
Let $\Qcal=\int_{S^1} \frac{K^2}{2}e^{2\eta}a^{-1}d\theta $. After integration in the spatial variable of the Einstein equations
\eqref{eq:lnalphar}-\eqref{weakconstraintsr1}, we obtain 
\begin{eqnarray}
\Pscr_R&= & \frac{\mathcal{Q}}{R^3}. \label{eq:P}\\
\Qcal_R &=& 2 R \mathcal{F} \Qcal \Pscr^{-2}+\Omega_\Qcal, \label{eq:Q}
\end{eqnarray}
where $\Omega_\Qscr$ is given by
$$
\Omega_{\Qcal}= 2 R \left( \int_{S^1} \frac{K^2}{2} E e^{2\eta}-\Pscr^{-1} \Escr \Qscr\right).
$$
As in Lemma \ref{lem:eee}, $\Omega_{\Qcal}$ satisfies the estimate
\begin{equation}\label{es:oqet}
|\Omega_{\Qcal}| \le R K^2 \Escr \int_{0}^{2\pi} \left( \frac{e^{2\eta}}{a}\right)_R d\theta=  2 R  \Escr \Qscr_R.
\end{equation} 

\subsubsection*{Renormalization}

According to our previous discussion, we expect $\Pscr$ to blow-up in the limit. One can check heuristically that "$\Pscr$ growing like $R^{1/2}$" and "$\Qscr$ growing like $R^{5/2}$" seem the only possibilities (as powers of $R$) compatible with the equations, under the assumption that $\Pscr \Escr$ behaves like $R^{-1}$ (see the discussion at the end of Section \ref{se:hgs}). Thus, one may try to introduce variables $\tilde{c}=\Pscr R^{-1/2}$ and $\tilde{d}=\Qscr R^{-5/2}$ and prove that $\tilde{c}$ and $\tilde{d}$ converge to some finite values. Using \eqref{eq:Q}, the equation for $\tilde{d}$ is then 
$$
\tilde{d}_R=\frac{\tilde{d}}{R} \left( 2 R^2 \Fcal \Pscr^{-2}-5/2 \right)+\Omega_\Qscr.
$$
From this equation, and the coupled equation for $\tilde{c}$, it is not clear whether $\tilde{c}$ and $\tilde{d}$ converge. However, assuming $\Omega_\Qscr$ to be a negligible term, it suggests that $2 R \Fcal \Pscr^{-2} \rightarrow 5/2$ as $R \rightarrow +\infty$. Equivalently, it suggests that $\frac{\Pscr}{R \Fcal^{1/2}} \rightarrow \frac{2}{\sqrt{5}}$. Similarly, one can guess that $\frac{\Qscr}{R^3 \Fcal^{1/2}} \rightarrow \frac{1}{\sqrt{5}}$. We thus introduce a new set of variables $c$ and $d$, replacing $\Pscr$ and $\Qscr$ based on these considerations.

However, since it is actually $\Gcal$ that satisfies ``good'' differential inequalities, we define $c,d$ as
\begin{eqnarray}
c : =\frac{\Pscr}{R\sqrt{\Gcal}} \\
d : =\frac{\Qscr}{R^3\sqrt{\Gcal}}
\end{eqnarray}
where we recall that $\Gcal=\Pscr\left(\Escr + \Gscr^U \right).$ 
Once again, we emphasize that while $\Gcal$ behaves asymptotically as $\Fcal$, 
it is important to use this normalization
rather than that of $\Fcal$, since the normalization procedure will introduce
a derivative of $\Gcal$ in the equation and it is this derivative (rather 
than the one of $\Fcal$) that we can control directly. 

 Note that while $\Fcal$ is manifestly non-negative, this is not the case for $\Gcal$. In the rest of this section, we will assume that $\Gcal > 0$, which ensures that all the computations below (as well as the definitions of $c$ and $d$) make sense. In the next section, a lower bound on $\Gcal$ using a bootstrap argument will be recovered.

An easy computation shows that $(c,d)$ satisfies 
\begin{eqnarray}
c'&=& \frac{d}{R}-\frac{c}{R}-\frac{c}{2} \frac{\Gcal'}{\Gcal}, \label{eq:c}\\
d'&=& \frac{\Fcal}{\Gcal}\frac{2dc^{-2}}{R}-\frac{3}{R}d-\frac{d}{2} \frac{\Gcal'}{\Gcal}
+ \frac{\Omega_{\Qcal}}{R^3 \sqrt{\Gcal}}. \label{eq:d}
\end{eqnarray}
To find the correct limits for $(c,d)$, let us first consider, the ordinary differential system 
\be
\aligned
c'&= \frac{d}{R}-\frac{c}{R}+\frac{c}{2R},
 \\
d'&= \frac{2dc^{-2}}{R}-\frac{3}{R}d+\frac{d}{2 R},  
\endaligned
\ee
which is obtained from the previous one by replacing $\frac{\Fcal}{\Gcal}$ by $1$, dropping the error term $\frac{\Omega_{\Qcal}}{R^3 \sqrt{\Gcal}}$ and replacing $-\frac{\Gcal'}{\Gcal}$ by $1/R$.

Looking now for a static point $(c_\infty, d_\infty)$ of the above system, we find that there is only one solution: $c_\infty=\frac{2}{\sqrt{5}}$, $d_\infty=\frac{1}{\sqrt{5}}$.
Thus, let us introduce $c_1$, $d_1$ by
\be
\aligned
& c_1 =c-\frac{2}{\sqrt{5}},
 \\
& d_1=d-\frac{1}{\sqrt{5}}
\endaligned
\ee
We finally deduce the equations satisfied by $c_1, d_1$ from the equations \eqref{eq:c}-\eqref{eq:d}, that is, 
\begin{eqnarray}
c_1' &=& \frac{d_1+\frac{1}{\sqrt{5}}}{R}-\frac{\frac{2}{\sqrt{5}}+c_1}{R}
-\frac{\frac{2}{\sqrt{5}}+c_1}{2} \frac{\Gcal'}{\Gcal}, \label{eq:c11}\\
d_1' &=& \frac{\Fcal}{\Gcal}\frac{2}{R} \frac{d_1+\frac{1}{\sqrt{5}}}{\left(\frac{2}{\sqrt{5}}
+c_1\right)^2}-\frac{3}{R}\left(d_1+\frac{1}{\sqrt{5}}\right)
-\left(d_1+\frac{1}{\sqrt{5}}\right)\frac{\Gcal'}{2\Gcal}+ \frac{\Omega_{\Qcal}}{R^3 \sqrt{\Gcal}}. \label{eq:d11}
\end{eqnarray}

Looking first at \eqref{eq:c11}, we rewrite it in the form
\begin{eqnarray*}
c_1'
 &=& \frac{1}{R}d_1-\frac{1}{2}\frac{c_1}{R}
-\frac{c_1}{2R}\left( 1 +R\frac{\Gcal'}{\Gcal}\right)  
-\frac{1}{R\sqrt{5}} \left( 1 +R\frac{\Gcal'}{\Gcal}\right).
\end{eqnarray*}
From \eqref{eq:d11}, elementary calculations (keeping in mind the linearization of the system) lead us to 
\begin{eqnarray*}
d_1' &=&  - \frac{5}{2R}c_1 +\frac{d_1}{R} 
\left( -\frac{1}{2}- \frac{\Gcal'}{2\Gcal}R \right) 
- \frac{1}{R} \frac{1}{2\sqrt{5}} \left( 1 + \frac{R\Gcal'}{\Gcal} \right) \\
&&+ \frac{1}{R\left(c_1+\frac{2}{\sqrt{5}} \right)^2 } f(d_1,c_1) \\ 
&&+\frac{2}{R}\left( \frac{\Fcal}{\Gcal}-1\right)\frac{d_1+\frac{1}{\sqrt{5}}}{\left(\frac{2}{\sqrt{5}}
+c_1\right)^2}
+ \frac{\Omega_{\Qcal}}{R^3 \sqrt{\Gcal}},
\end{eqnarray*}
where $f(c_1,d_1)$ is a polynomial in $c_1$ and $d_1$ with vanishing linear part (the first terms are quadratic in $c_1$ and $d_1$).
Thus, we have
\begin{eqnarray}
d_1' &=&-\frac{5}{2R}c_1 + \Omega_{lin}^d+\Omega_1^d+\Omega_2^d +\Omega_3^d+\Omega_4^d, 
\\
c_1' &=&  \frac{d_1}{R}-\frac{c_1}{2R}+\Omega_{lin}^c+\Omega_1^c.
\end{eqnarray}
where the terms $\Omega_i^{c,d}$ contain all the error terms, i.e.
\begin{eqnarray}
\Omega_{lin}^d&=& -\frac{d_1}{2R} \left( 1+ \frac{\Gcal'}{\Gcal}R \right),
 \label{eq:oedlin} \\
\Omega_1^d&=&-\frac{1}{R} \frac{1}{2\sqrt{5}} \left( 1 + \frac{R\Gcal'}{\Gcal} \right),
\\
\Omega_2^d&=&\frac{1}{R\left(c_1+\frac{2}{\sqrt{5}} \right)^2 }  f(d_1,c_1),
\end{eqnarray}
\begin{eqnarray}
\Omega_3^d&=&\frac{2}{R}\left( \frac{\Fcal}{\Gcal}-1\right)\frac{d_1+\frac{1}{\sqrt{5}}}{\left(\frac{2}{\sqrt{5}}
+c_1\right)^2},
\\
\Omega_4^d&=& \frac{\Omega_{\Qcal}}{R^3 \sqrt{\Gcal}},
 \label{eq:od4}
 \end{eqnarray}
\begin{eqnarray}
\Omega_{lin}^c&=&- \frac{c_1}{2R}( 1+ \frac{R\Gcal'}{\Gcal}),
 \\
\Omega_1^c&=&  -\frac{1}{R\sqrt{5}} ( 1 +R\frac{\Gcal'}{\Gcal}).
 \label{eq:oe1c}
\end{eqnarray}

Setting now $u:=\left(\begin{array}{c}c_1\\d_1\end{array} \right)$, we rewrite the system under consideration as 
$$
u' = 1/R\left( 
\left( \begin{array}{cc} 
-1/2 & 1 \\
-5/2 & 0 
\end{array}
\right)-\frac{1}{2}\left( 1+ \frac{\Gcal'R}{\Gcal}\right)I_2\right)u
+ \omega, 
$$
where $\omega$ contains all the terms $\Omega_i^{c,d}$ apart from $\Omega^d_{lin}$ and $\Omega^c_{lin}$, and $I_2$ denotes the identity matrix. Consider the matrix 
$$
A=\left( \begin{array}{cc} 
-1/2 & 1 \\
-5/2 & 0 
\end{array}
\right)
$$
and also let  
$$
B=-\frac{1}{2}\left(1+ \frac{\Gcal'R}{\Gcal}\right)I_2.
$$
Then, we find 
\be
u= \exp\int_{R_0}^{R} \frac{A+B}{R'}dR'\, u(R_0)
+\int^R_{R_0}\left[\exp\int_{R'}^R  \frac{A+B}{R''}dR'' \right]\omega(R') dR'.
\ee

Note next that 
$$
\aligned
\exp\int_{R_0}^{R} \frac{A+B}{R'}dR'
& =\exp\int^R_{R_0} \frac{A}{R'}dR'\exp\int^R_{R_0} \frac{B}{R'}dR'
\\
& =\exp\int^R_{R_0} \frac{A}{R'}dR'\left( \frac{R_0 \Gcal(R_0) }{R\Gcal(R)}\right)^{1/2}
\endaligned
$$
and that the eigenvalues of $A$ are 
$\lambda_{\pm}= \frac{-1}{4} \pm \frac{i\sqrt{39}}{4}$. 
Hence  $$||\exp\int^R_{R_0} \frac{A}{R'}dR'|| 
\le C_A \left(\frac{R_0}{R}\right)^{1/4},$$ for some constant $C_A > 0$ depending on the matrix $A$ and we have the following result. 

\begin{proposition} Provided the corrected energy $\Gcal$ is positive for all $R \in [R_0,R_1]$, one has, for all $R \in [R_0,R_1],$
\begin{eqnarray} \label{es:u}
|u(R)|& \le&  C_A\left(\frac{R_0}{R}\right)^{1/4}\left( \frac{R_0 \Gcal(R_0) }{R\Gcal(R)}\right)^{1/2}|u(R_0)| \\
&&\hbox{}+\int^R_{R_0}  C_A \left(\frac{R'}{R}\right)^{1/4}\left( \frac{R'\Gcal(R') }{R\Gcal(R)}\right)^{1/2}|\omega(R')|dR', \nonumber
\end{eqnarray}
where 
$$
|\omega| \le C\left(|\Omega^d_1|+|\Omega^d_2|+|\Omega^d_3|+|\Omega^d_4|+|\Omega^c_1|\right).
$$
\end{proposition}
It remains to combine the above inequality with our differential inequalities for $\Gcal$ and estimates on the error terms. 


\subsection{Source-terms of the dynamical system} 

We now combine our results in the latter two sections and we estimate the source-terms of the dynamical system. We will assume here that $\Gcal$ is strictly positive, a property that we shall retrieve below in a bootstrap argument. 

\subsubsection*{Estimate for $\left|\frac{\Omega_{\Qcal}}{R^3 \sqrt{\Gcal}}\right|$}

Since we have
\begin{eqnarray*}
\Qcal &=& dR^3 \sqrt{\Gcal},
 \\
\Qcal_R &=& d_R R^3 \sqrt{\Gcal} + 3 d  R^2 \sqrt{\Gcal}+ d \frac{R^3 }{2} \frac{\Gcal'}{\sqrt{\Gcal}},
\end{eqnarray*}
it follows that
$$
\left|\frac{\Omega_{\Qcal}}{R^3 \sqrt{\Gcal}}\right| 
\le 2 \Escr R d_R + 6 R \Escr \frac{d}{R}+  R \Escr d  \frac{\Gcal'}{\Gcal}.
$$
Observe that, while some terms in the right-hand side have no sign, their sum does (because $\mathcal{Q}_R$ is positive). 

\subsubsection*{Estimating $R\Gcal'/\Gcal+1$}

From the corrected energy estimate, we get 
\begin{eqnarray*}
\left| \frac{\Gcal'}{\Pscr}+ \frac{\Gcal}{\Pscr R} \right| \le \frac{\Pscr}{4 \pi R^2} \Escr+
\frac{\mathcal{A}}{R^3  \Pscr(R) } + \Escr \frac{K^2}{2R^3}
\left( \int_{S^1} \frac{e^{2\eta}}{a}\right)_R+ \frac{\Pscr_R}{R}\Escr, 
\end{eqnarray*}
 hence
\begin{eqnarray}
\left| \frac{R\Gcal'}{\Gcal}+1 \right| &\le& \frac{ \Pscr }{4\pi R}\frac{\Fcal}{\Gcal} +\frac{\mathcal{A}}{\Gcal R^{2}}
+ \frac{\mathcal{F}}{\mathcal{G}} \frac{K^2}{2R^2}
\left( \int_{S^1} \frac{e^{2\eta}}{a}\right)_R
+\frac{\mathcal{F}}{\Gcal} \Pscr_R \nonumber 
\\
 &\le&\frac{\mathcal{A}}{ \Gcal R^2}
+ \frac{\mathcal{F}}{\mathcal{G}} \frac{\Qcal_R}{R^2}
+\frac{\mathcal{F}}{\Gcal} \frac{\sqrt{\Gcal}}{4\pi} c
+\frac{\mathcal{F}}{\Gcal} \sqrt{\Gcal} d. \label{es:Rggp}
\end{eqnarray} 

\subsubsection*{Estimates for $\Omega_1^i$}

It follows from the estimate \eqref{es:Rggp} and the definition of $\Omega_1^c$ and $\Omega_1^d$ that there exists a  constant
$C> 0$ such that, for $i=d,c$:
\begin{eqnarray}\label{es:oi1}
| \Omega_1^i | \le \frac{C}{R} \left( \frac{\mathcal{A}}{\Gcal R^2}
+ \frac{\mathcal{F}}{\mathcal{G}} \frac{\Qcal_R}{R^2}
+\frac{\mathcal{F}}{\Gcal} \frac{\sqrt{\Gcal}}{4\pi} c
+\frac{\mathcal{F}}{\Gcal} \sqrt{\Gcal} d \right).
\end{eqnarray}

\subsubsection*{Estimates for $\Fcal \Gcal^{-1}$ and $\Omega_3^d$}

Using Lemma \eqref{lem:egu}, we have
\begin{eqnarray}\label{es:fg}
\left| \frac{\Fcal}{\Gcal}-1\right|= \left| \frac{\Fcal-\Gcal}{\Gcal} \right|
= \left| \frac{\Pscr \Gscr^U }{\Gcal}\right| 
\le \frac{1}{4\pi R}\frac{\Pscr^2 \Escr}{\Gcal} \le \frac{\Pscr}{4\pi R} \frac{\Fcal}{\Gcal}.
\end{eqnarray}
As a consequence, provided that $c_1$ is sufficiently small so that $\frac{2}{\sqrt{5}}+c_1$ is bounded from below by, say $\frac{1}{\sqrt{5}}$, we find 
\begin{eqnarray} \label{es:od3}
|\Omega^d_3| \le C \left( |d_1| +1 \right)\frac{\Pscr}{4\pi R^2} \frac{\Fcal}{\Gcal} 
\end{eqnarray}
for some constant $C > 0$.

Note that at this point, we have estimates on all the error terms arising in \eqref{eq:oedlin}-\eqref{eq:oe1c}, apart from $\Omega_2^d$ which will be estimated directly in the next section (using a smallness assumption on $c, d$).

\subsubsection*{ Estimates on $\Gcal$}

After integration of the corrected energy estimate, we find 
\begin{eqnarray}
| R\Gcal -R_0 \Gcal(R_0) | &\le& \mathcal{A} \left( \frac{1}{R_0}-\frac{1}{R}\right)
+\int^R_{R_0} \left( \frac{\Pscr}{4 \pi R' } \Fcal+ \frac{\Fcal \Qcal_R}{R'^2}+ \frac{\Fcal \Qcal}{R'^3} \right) dR'.
\end{eqnarray}
The last term can be rewritten in terms of $\Pscr_R$, giving  
\begin{eqnarray}
| R\Gcal -R_0 \Gcal(R_0) | \le  \label{es:gc} \mathcal{A} \left( \frac{1}{R_0}-\frac{1}{R}\right)
+\int^R_{R_0} \left( \frac{\Pscr}{4 \pi R' } \Fcal+ \frac{\Fcal \Qcal_R}{R'^2}+ \Fcal \Pscr_R \right)dR'.
\end{eqnarray}


\section{Small data theory} 
\label{se:sdp}

\subsection{Assumption on the initial data}

We now restrict ourselves to small data in the following sense. 
Fix $C_1 > 0$, $\mathcal{A} \in [0,+\infty)$, and $R_0> 0$, as well as some $\epsilon > 0$. 
Consider the class of initial data satisfying  
\begin{eqnarray}
R_0 \Gcal(R_0)-\frac{\mathcal{A}}{R_0} &\ge& C_1 > 0,
 \label{ineq:gah}\\
|c_1|(R_0) &\le& \epsilon, \label{ida:c} \\
|d_1|(R_0) &\le& \epsilon,\label{ida:d}\\
|\frac{\Fcal}{\Gcal}-1|(R_0)&\le& 1,\label{ida:fog} \\
\Gcal(R_0)+\frac{\mathcal{A}}{R_0^2} &\le& \epsilon, \label{ida:g} 
\qquad 
\label{sm:g}
\end{eqnarray}
where $\mathcal{A} = R_0^2\left( \int_{S^1} a^{-1} U_R \right)^2 (R_0)$.

Note that the first assumption implies in particular that $\Gcal > 0$.
The second and third assumptions imply that $\Pscr$ and $\Pscr_R$ are close to their expected asymptotic behavior (which depends on $\Escr$, hence the need for normalized quantities).
The fourth condition implies that the correction term $\Gscr^U$ is "not too large" compared to the energy $\Escr$. 
The last inequality  means that the (rescaled) energy is small.  

Let $R_b$ be the largest time $R$ such that the following bootstrap assumptions are valid in $\mathcal{B}:=[R_0,R_b)$. 
For all $R \in \mathcal{B}$, we have 
\begin{eqnarray}
|c_1|(R) < \epsilon^{1/4}, \label{bs:c1}
\\
|d_1|(R) < \epsilon^{1/4}, \label{bs:d1}
\\
|\frac{\Fcal}{\Gcal}-1|(R)< 2, \label{bs:gf}
\\
0 < \Gcal(R_0) < \left( R_0 \Gcal(R_0)+\frac{\mathcal{A}}{R_0} \right) \frac{2}{R}. \label{bs:g}
\end{eqnarray}
The set $\mathcal{B}$ is clearly open in $[R_0,+\infty)$.
Moreover, from the smallness assumptions it follows that $\mathcal{B}$ 
is also non-empty. 

As an immediate consequence of \eqref{bs:c1} and \eqref{bs:d1}, if $\epsilon$ is sufficiently small, then we have in $\mathcal{B}$
\begin{equation}
\frac{1}{c^2}(R_0)=\frac{1}{\left(c_1+\frac{2}{ \sqrt{5}}\right)^2}(R_0) 
\le 2,
\label{bs:c}
\end{equation} 
\begin{eqnarray}
|c| = |\frac{2}{\sqrt{5}}+c_1| &\le&  1, \label{bs:ca} \\
|d| = |\frac{1}{\sqrt{5}}+d_1|&\le&  1.
\end{eqnarray}
Furthermore, from \eqref{bs:g} and \eqref{ida:g}, we have immediately in $\mathcal{B}$,
\begin{equation} \label{ineq:gcalim}
\Gcal \le 2 \frac{R_0}{R} \left(\Gcal(R_0)+ \frac{\mathcal{A}}{R_0^2}\right) \le 2 \epsilon\frac{R_0}{R}\le 2 \epsilon.
\end{equation}

We now consider $C_1$ and $\mathcal{A}$ as fixed in \eqref{ineq:gah}. We will show that there exists an $\epsilon_0 > 0$ and a constant $r > 0$ such that for all $0<\epsilon< \epsilon_0$ and $R_0 > r$, the set 
$\mathcal{B}$ is closed; this will be done by "improving" each of the bootstrap assumptions \eqref{bs:c1}-\eqref{bs:g}. Moreover, $\epsilon_0$ will depend only on a lower bound for $r$ (as well as $\mathcal{A}$ and $C_1$).


\subsection{Improving the assumption on $\Fcal \Gcal^{-1}$}

In view of the estimate \eqref{es:fg}, we have 
\begin{eqnarray}
\left|\frac{\Fcal}{\Gcal}-1\right| \le 3 \frac{c \Gcal^{1/2}}{4\pi} \le \frac{3 \sqrt{2}}{4\pi}\epsilon^{1/2},  \label{ineq:impfg}
\end{eqnarray}
by using the bootstrap assumptions \eqref{bs:gf}, \eqref{bs:ca} and using \eqref{ineq:gcalim}. 
This improves \eqref{bs:gf}. 

Throughout, the letter $C$ will be used to denote numerical constants that are independent of $\epsilon$ and $R_0$ and may change at each occurrence. Thus, the above estimate reads 
$$
\left|\frac{\Fcal}{\Gcal}-1\right| \le C \epsilon^{1/2}.
$$

\subsubsection*{Improving the $\Gcal$ assumption}

From the corrected energy estimate \eqref{es:gc}, we have
\begin{eqnarray}
R\Gcal \le R_0 \Gcal(R_0) + \frac{\mathcal{A}}{R_0} \nonumber
+ \int^R_{R_0} R'\Gcal  \frac{\Fcal}{\Gcal} \left(\frac{\Pscr}{4 \pi R'^2}+ \frac{\Qcal_R }{R'^3}
+\frac{\Pscr_R}{R'}\right) dR', 
\end{eqnarray}
hence 
\begin{eqnarray} \label{es:gcalf}
\Gcal \le \frac{D_0}{R} \exp \int^R_{R_0} (1+C\epsilon^{1/2})\nonumber
\left[\frac{\Pscr}{4 \pi R'^2}+ \left(\Qcal_R R'^{-3}+\frac{\Pscr_R}{R'}\right) \right], 
\end{eqnarray}
where $D_0=  R_0 \Gcal(R_0) + \frac{\mathcal{A}}{R_0}$ 
and we have used the improved inequality \eqref{ineq:impfg}.

The integral $\int_{R_0}^R \frac{\Pscr}{4 \pi R'^2} dR'$ can be estimated using \eqref{ineq:gcalim}: 
\begin{eqnarray*}
\int_{R_0}^R \frac{\Pscr}{4 \pi R'^2} dR'&=&\int_{R_0}^R \frac{c R' \Gcal^{1/2}}{4 \pi R'^2} dR',\\
&\le&\int_{R_0}^R \frac{C \epsilon^{1/2}R_0^{1/2} }{4 \pi R'^{3/2}} dR' \le C \epsilon^{1/2}
\end{eqnarray*}
for some fixed numerical constant $C > 0$.

For the other integrals, we integrate by parts: 
\begin{eqnarray*}
\int^R_{R_0} \left[\left(\frac{\Qscr_R}{R'^3}+\frac{\Pscr_R}{R'}\right) \right]dR'&\le& \frac{\Qcal}{R^3}
+\frac{\Pscr}{R}+\int_{R_0}^R \left[\frac{3 \Qcal}{R'^4}+\frac{\Pscr}{R'^2}\right]dR'
 \\
&\le& (c+d) \Gcal^{1/2}+\int_{R_0}^R \frac{3d+c}{R'} \Gcal^{1/2}(R')dR'
 \\
&\le& C \epsilon^{1/2}+C \int_{R_0}^R\frac{R_0^{1/2}}{R'^{3/2}}\epsilon^{1/2}dR'
\le C \epsilon^{1/2}.
\end{eqnarray*} 
Combining this result with the previous estimate, we have thus obtained
\begin{eqnarray} \label{es:glb}
R\Gcal \le D_0 \exp\left( \left(1+ C\epsilon^{1/2} \right)C \epsilon^{1/2} \right) < 3/2 D_0,
\end{eqnarray}
providing that $\epsilon$ is small enough. This improves \eqref{bs:g}.


\subsubsection*{A lower bound on $\Gcal$}

We derive here a lower bound on $R\Gcal$. From the corrected energy inequality in differential form \eqref{dineq:glb} and the estimates on the error term we have
\begin{eqnarray} \label{dineq:glbf}
\frac{d}{dR} \left( R \Gcal \right) \ge -\frac{\Acal}{R^2}-R\Gcal \left[  \frac{\Fcal}{\Gcal} \left(\frac{\Pscr}{4 \pi R'^2}+ \frac{\Qcal_R}{R'^3}
+\frac{\Pscr_R}{R'}\right) \right].
\end{eqnarray}
Let $$\Omega'= \frac{\Fcal}{\Gcal} \left(\frac{\Pscr}{4 \pi R'^2}+ \frac{\Qcal_R}{R'^3}
+\frac{\Pscr_R}{R'}\right).$$
The estimates of the previous sections have shown that $$\int_{R_0}^R \Omega' dR' \le C \epsilon^{1/2}.$$
We can rewrite \eqref{dineq:glbf} as
$$
\frac{d}{dR} \left( R \Gcal \right) \ge -\frac{\Acal}{R^2} -R\Gcal \Omega'
$$
leading to
\begin{eqnarray*}
\frac{d}{dR} \left( R \Gcal \exp{\int_{R_0}^R \Omega' dR'} \right) 
&\ge& -\frac{\Acal}{R^2} \exp{\int_{R_0}^R \Omega' dR'}
\\
&=& \frac{d}{dR} \left( \frac{\Acal}{R} \right) \exp{\int_{R_0}^R \Omega' dR'},\\
&=&\frac{d}{dR} \left( \frac{\Acal}{R}  \exp{\int_{R_0}^R \Omega' dR'}\right)- \frac{\Acal}{R} \Omega'\exp{\int_{R_0}^R \Omega' dR'}.
\end{eqnarray*}

Thus,
\begin{eqnarray*}
\frac{d}{dR} \left[ \left( R \Gcal -\frac{\mathcal{A}}{R}\right)\exp{ \int_{R_0}^R \Omega' dR'} \right] \ge -\frac{\Acal}{R} \Omega' \exp{ \int_{R_0}^R \Omega' dR'},
\end{eqnarray*} 
which leads after integration to 
\begin{equation}
\label{es:lbg}
R \Gcal - \frac{\Acal}{R} \ge \left( R_0\Gcal(R_0)-\frac{\mathcal{A}}{R_0} \right)(1-C \epsilon^{1/2}) - \frac{\Acal}{R_0}C \epsilon^{1/2} =C_1(1-C \epsilon^{1/2}) - \frac{\Acal}{R_0}C \epsilon^{1/2} \ge \frac{C_1}{2}.
\end{equation}
provided that $\epsilon$ is sufficiently small depending on $\Acal$, $C_1$ and a lower bound on $R_0$.

Since $\Acal \ge 0$, we have thus obtained $R \Gcal\ge \frac{C_1}{2}$. In particular, we have improved the lower bound bootstrap inequality for $\Gcal$. 
\begin{remark}
Instead of starting from the corrected energy inequality in differential form, one could use here the estimate \eqref{es:gc} as well as the estimates of the previous section to estimate the term containing $\Gcal$ in the error term. This would lead to an estimate of the form
$$
R \Gcal \ge C_1 - D_0 C\epsilon^{1/2}
$$
and would therefore require $\epsilon$ to be small compared to $D_0$. The above method has the advantage of not constraining $\epsilon$ any further.
\end{remark} 

\subsubsection*{Improving the $c_1$, $d_1$ assumptions}

Using the lower bound on $\Gcal$ just obtained, the bootstrap assumption \eqref{bs:g}, the initial data assumptions \eqref{ida:c} and \eqref{ida:d} and the fact that $\frac{R'}{R} \le 1$ if $R' \in [R_0,R]$, it follows from \eqref{es:u} that
\be
\aligned
\label{ineq:ub}
|u| 
& \le \left(\frac{ C_A D_0}{C_1}\right)^{1/2} \epsilon +  C \left(\frac{4 D_0}{C_1}\right)^{1/2}\int^R_{R_0}  \left( |\Omega^c_1|+|\Omega^d_1|+|\Omega^d_4|\right) dR'
\\
& +C \left(\frac{4 D_0}{C_1}\right)^{1/2}\int^R_{R_0} \left(\frac{R'}{R}\right)^{1/4} \left(|\Omega_2^d|+|\Omega_3^d| \right) dR'.
\endaligned
\ee

We now estimate all the error terms in $\omega$. 
First, we have
\begin{eqnarray}
|\Omega^c_1, \Omega^d_1| &\le& \frac{C}{R} \left| 1+ R\frac{\Gcal'}{\Gcal} \right| \nonumber \\
&\le& \frac{C}{R} \left( \frac{ 2}{C_1} \frac{\mathcal{A}}{R} + C \frac{\Qscr_R}{R^2} + C \Gcal^{1/2}\right), \label{ineq:o1e}
\end{eqnarray}
using \eqref{es:Rggp}, \eqref{es:oi1} and \eqref{bs:gf}.
The first term in the parentheses in the right-hand side of the last inequality will contribute to \eqref{ineq:ub} as
\begin{eqnarray*}
\left(\frac{4 D_0}{C_1}\right)^{1/2}\int^R_{R_0}\frac{ 2}{C_1} \frac{\mathcal{A}}{R'^2}dR' &\le& C \frac{\mathcal{A}}{C_1^{3/2}} D_0^{1/2} R_0^{-1} \nonumber \\
&\le& C\frac{\mathcal{A}}{C_1^{3/2} R_0^{1/2}}(D_0 R_0^{-1})^{1/2} \nonumber \\
&\le& C(C_1,R_0, \mathcal{A})\epsilon^{1/2},
\nonumber  
\end{eqnarray*}
by using the smallness assumption \eqref{ida:g}. The second term can be estimated using an integration by parts leading to the estimate
$$
C \left( \frac{D_0}{C_1} \right)^{1/2} \int_{R_0}^R \frac{\Qcal_R}{R'^3}dR' \le  C\left( \frac{D_0}{C_1} \right)^{1/2}\epsilon^{1/2}.
$$
Since  $\frac{D_0}{C_1}=1-2\frac{\Acal}{C_1 R_0}$, we thus obtain 
$$
C \left( \frac{D_0}{C_1} \right)^{1/2} \int_{R_0}^R \frac{\Qcal_R}{R'^3}dR' \le C \epsilon^{1/2},
$$
by choosing $\epsilon$ sufficiently small depending only on a lower bound on $C_1$, $\mathcal{A}$ and a lower bound on $R_0$.

The last term in \eqref{ineq:o1e} can be estimated using \eqref{ineq:gcalim} leading to 
$$
\int_{R_0}^R \frac{\Gcal^{1/2}}{R'}dR' \le C \epsilon^{1/2}.
$$

The estimates for $\Omega_2^d, \Omega_3^d$ are straightforward using the bootstrap assumptions
\begin{eqnarray} 
|\Omega^d_2| &\le& \frac{C}{R}\epsilon^{1/2},\nonumber \\
|\Omega^d_3| &\le& \frac{C}{R}  \epsilon^{1/2}.
\nonumber
\end{eqnarray}
For $\Omega_4^d$, we note that in view of \eqref{eq:od4} and \eqref{es:oqet}, we have
$$
|\Omega^d_4| \le \frac{2R \Escr \Qcal_R}{R^3 \sqrt{\Gcal}}.
$$
Then, we note that 
$$
\Escr = \frac{\Fcal}{\Pscr}= \frac{\Fcal}{c R \Gcal^{1/2} },
\nonumber
$$
hence
$$
\Escr \Gcal^{-1/2}=  \frac{1}{c R}  \left( \frac{\Fcal} {\Gcal} \right).
$$
Using the bootstrap assumptions, this leads to 
\begin{equation}\label{es:od4}
|\Omega^d_4 |\le   \frac{1}{c}\frac{\Fcal }{\Gcal}  \frac{2 \Qcal_R}{R^3} \le C \Qcal_R R^{-3},
\end{equation}
where we have used that $\Qcal_R \ge 0$ in the last estimate. Its integral can then be estimated by integration by parts, as we have already done previously.

Combining all these estimates leads us to 
\begin{eqnarray*}
|u| &\le& C \left(\frac{2 D_0}{C_1}\right)^{1/2} \epsilon + C(C_1,R_0, \Acal) \epsilon^{1/2} \\
&\le& C \left(\mathcal{A}, R_0, C_1 \right) \epsilon^{1/2},
\end{eqnarray*}
which improves \eqref{bs:c1} and \eqref{bs:d1}.
In conclusion, we have improved all of the bootstrap inequalities 
and it follows that 
\be \nonumber
\mathcal{B}=[R_0,+\infty).
\ee


\section{The asymptotic regime} \label{se:ar}

In this section, we state and prove our main result. 

\begin{theorem}[Late-time asymptotics of $T^2$-symmetric polarized vacuum spacetimes]
 \label{th:mt}Let $\mathcal{A} \ge 0$ and $C_1 > 0$ and $r > 0$ be fixed constants. Then, there exists an $\epsilon_0$ such that if $0 \le \epsilon \le \epsilon_0$ and $R_0 \ge r$, for any initial data set satisfying the smallness conditions \eqref{ineq:gah}-\eqref{sm:g},    
the associated solution has the following asymptotic behavior: for all times $R \ge R_0$ and all $\theta \in S^1$,
\begin{eqnarray}
|u|(R,\theta) &=& O(R^{-1/4}), \label{ab:u} \\
\left|R\Gcal(R) - C_\infty \right| &=& O(R^{-1/2}),  \label{ab:g}\\
\left| \Pscr(R) -\frac{2}{\sqrt{5}} C_\infty^{1/2} R^{1/2}\right|&=&O(R^{1/4}),  \label{ab:p}
\end{eqnarray}
\begin{eqnarray}
\left|\Qscr(R) - \frac{1}{\sqrt{5}} C_\infty^{1/2} R^{5/2}\right|&=&O(R^{9/4}), \label{ab:q}\\
\left| \Escr(R)-\frac{\sqrt{5}C_\infty^{1/2}}{2 R^{3/2}} \right| &=& O(R^{-7/4}) \label{ab:e}, 
\\
\left|\frac{1}{2\pi}\int_{S^1}\eta(R,\theta') d\theta'-\eta(R,\theta)\right| &=& O( R^{-1/2}), \label{ac:e}
\end{eqnarray}
\begin{eqnarray}
\left|K^2 e^{2\eta}(R,\theta) -R^2\right|&=&O(R^{7/4}), \label{ab:et}
\\
\left| a^{-1}(R,\theta) \Pscr^{-1}(R)-\mathcal{L}(\theta) \right|&=&O(R^{-1/2}),\label{ab:a} \\
\left| \frac{1}{2\pi}\int_{S^1}U(R,\theta) d\theta-U(R,\theta)\right|&=& O(R^{-1/2}), \label{ab:Ut}
\end{eqnarray}
\begin{eqnarray}
\left|U(R,\theta)-C_U\right|&=&O(R^{-1/2}), \label{ab:U} \\
\left| H(R,\theta) - \frac{4}{K\sqrt{5}}C_\infty^{1/2}R^{1/2}\mathcal{L}(\theta) \right|&=&O(R^{1/4}),\label{ab:H}
\end{eqnarray}
where $C_\infty>0$ and $C_U$ are constants depending on the solution  and $\mathcal{L}(\theta)$ is  a $W^{1,1}(S^1)$ strictly positive function. 
\end{theorem}

\begin{proof}
Most of the above estimates are simply obtained by revisiting the proof in the previous section and checking that the error terms are now {\sl integrable.} 

For instance, in order to prove \eqref{ab:u}, note that from \eqref{es:u} and the  estimates of Section \ref{se:sdp}, we have
\begin{equation} \label{es:uar}
|u| \le C R^{-1/4} \left(1+ \int_{R_0}^R R'^{1/4} |\omega(R')|dR' \right).
\end{equation}
From \eqref{ineq:o1e} and \eqref{es:od4}, one can easily see that the contributions of $\Omega_{1}^c$, $\Omega_1^d$ and $\Omega_{4}^d$ are integrable in $R$. For instance, using an integration by parts,
\begin{eqnarray*}
\int_{R_0}^R \frac{\Qcal_R}{R'^{3-1/4}} &\le& C \frac{\Qcal}{R^{3-1/4}}+C \int_{R_0}^R  \frac{\Qcal}{R'^{4-1/4}} dR', \\
 &\le& C \frac{\Qcal}{R^3 \Gcal^{1/2}} (R\Gcal)^{1/2} R^{-1/4} + C \int_{R_0}^R\frac{\Qcal}{R'^3 \Gcal^{1/2}} (R'\Gcal)^{1/2}R'^{-5/4} dR', \\
&\le & C  R^{-1/4} + C \int_{R_0}^RR'^{-5/4} dR' \le C.
\end{eqnarray*}
For $\Omega_3^d$, it follows from \eqref{es:od3} and the estimates of the previous section that
$| \Omega_3^d | \le C R^{-3/2}$. Thus, its contribution to the integral of \eqref{es:uar} is integrable.
Since moreover, $|\Omega^d_2| \le \frac{C}{R} |u|^2$, \eqref{es:uar} has now been reduced to 
\begin{eqnarray}\label{es:fglu2}
|u| \le C R^{-1/4} \left( 1+ \int_{R_0}^R R'^{-3/4} |u|^2 (R')dR' \right).
\end{eqnarray}
Since we already know from the estimates of the previous section that  $|u| \le C \epsilon^{1/2}$, an application of Gronwall's lemma gives us the weak bound
$$
R^{1/4} |u| \le C R^{\epsilon^{1/2}}. 
$$
It then follows that $R^{-3/4} |u|^2 \le C R^{-5/4+\epsilon}$ and thus, for $\epsilon$ sufficiently small, \eqref{es:fglu2} now implies the desired estimate \eqref{ab:u}.

Similarly, to prove \eqref{ab:g}, first note that $\frac{d}{dR} \left( R\Gcal \right)$ is integrable, using the estimates of Section \ref{se:sdp} and \eqref{es:Rggp}. Thus, there exists a constant $C_\infty$ such that $R \Gcal \rightarrow C_\infty$, as $R \rightarrow +\infty$. Since $R \Gcal$ is uniformly bounded from below in view of \eqref{es:lbg}, we have $C_\infty> 0$. To get the rate of convergence, it then suffices to write $R \Gcal-C_\infty=\int_R^\infty \frac{d}{dR} \left( R' \Gcal \right) dR'$ and to estimate the integral as before.

Then, \eqref{ab:p}, \eqref{ab:q} and \eqref{ab:e} follow from the definitions of $\Pscr$, $\Qcal$ and $\Escr$.

For \eqref{ac:e}, using \eqref{weakconstraintsr2}, the simple estimate $F \le E$ and \eqref{ab:e}, we have, for all $R \ge R_0$ and $\theta \in S^1$, 
\begin{eqnarray*}
\left|\frac{1}{2\pi}\int_{S^1}\eta(R,\theta') d\theta'-\eta(R,\theta)\right| &\le& \int_{S^1} |\eta_\theta| (R,\theta')d\theta', \\
&\le& \int_{S^1} R F d\theta' \le \int_{S^1} R E \le C R^{-1/2}
\end{eqnarray*}
for some $C > 0$.
For \eqref{ab:et}, we use \eqref{ac:e}, \eqref{ab:p}, \eqref{ab:q} as well as
\begin{eqnarray*}
\Pscr \frac{K^2}{2} e^{2\eta}(R,\theta) &=& \int_{S^1} a^{-1}(R,\theta') \frac{K^2}{2} e^{2\left(\eta(R,\theta)-\eta(R,\theta')+\eta(R,\theta') \right)}d\theta', \\
&=& \int_{S^1} a^{-1}(R,\theta') \frac{K^2}{2} e^{2\left(\eta(R,\theta')+O(R^{-1/2}) \right)}d\theta'\\
&=&\Qcal\left(1+O(R^{-1/2}) \right).
\end{eqnarray*}

For \eqref{ab:a}, we first differentiate \eqref{eq:lnalphar} in $\theta$, that is, 
\begin{eqnarray}\label{eq:art}
\left( 2 \ln a \right)_{R \theta}=- \frac{K^2}{R^3}e^{2\eta} 2 \eta_\theta.
\end{eqnarray}
Note that the right-hand side is integrable in $L\left([R_0,+\infty) \times S^1 \right)$ since
\begin{eqnarray}\label{es:ilnart}
\int_{R_0}^\infty \int_{S^1}\left| \frac{K^2}{R^3}e^{2\eta} 2 \eta_\theta \right|d\theta dR \le \int_{R_0}^R C R^{-1} R \Escr \le C,
\end{eqnarray}
in view of \eqref{ab:e}.
This implies that $\left( \ln a \right)_{\theta}(R,\theta)$ converges in $L^1(S^1)$ as $R \rightarrow +\infty$ to some function $\mathcal{R}(\theta) \in L^1(S^1)$ and, moreover, we have the estimate
$$
|| \left( \ln a \right)_{\theta} -\mathcal{R} ||_{L^1(S^1)} = O(R^{-1/2}),
$$ by using \eqref{es:ilnart}.

Integrating over $[\theta,\theta']$, we get 
$$
\frac{a(R,\theta)}{a(R,\theta')} = \exp\left({\int_{\theta'}^\theta \mathcal{R}(\theta'') d \theta'' + O(R^{-1/2})}\right).
$$
Integrating again in the $\theta'$ variable, we get
$$
\left|a(R,\theta)\mathcal{P}- \int_{S^1} e^{\int_{\theta'}^\theta \mathcal{R}(\theta'') d \theta''} d\theta' \right| \le C \left( \exp{\left(O(R^{-1/2})\right)}-1 \right)=O(R^{-1/2}).
$$

For \eqref{ab:H}, it is sufficient to note that with the knowledge of the asymptotic behavior of $a$ and $\eta$ and equation \eqref{GHequa}, we can integrate $H_R$ directly and then compute the integral up to some error.

The property \eqref{ab:Ut} is an easy consequence of \eqref{ab:a}, \eqref{ab:p} and  \eqref{ab:e}.
For \eqref{ab:U}, we observe that 
$$
\left| \frac{d}{dR} \int_{0}^{2\pi} U d\theta \right|= \left| \int_0^{2\pi} U_R d\theta \right|  \le (2\pi)^{1/2} \left( \int_0^{2\pi} U_R^2 d\theta \right)^{1/2}
$$
and 
\begin{eqnarray*}
\left( \int_0^{2\pi} U_R^2 d\theta\right)(R) =\left(\int_0^{2\pi}a^{-1} a U_R^2 d\theta\right)(R) &\le& \sup_{[0, 2\pi]} a (R,\theta) \int_0^{2\pi}a^{-1} U_R^2 d\theta,\\
&\le& \left( \frac{1}{\Pscr}+o(a)\right)\frac{1}{\mathcal{L}(\theta)}\int_0^{2\pi}a^{-1} U_R^2 d\theta, \\
&\le&  \frac{C}{R^2}
\end{eqnarray*}
for some $C > 0$. Here, we have used \eqref{ab:a}, together with the fact $\mathcal{L}$ is bounded away from zero uniformly, as well as \eqref{ab:p} and \eqref{ab:e}. 

This implies that 
$$
\left| \frac{d}{dR} \int_{0}^{2\pi} U d\theta \right| \le \frac{C}{R}
$$
and by integration and \eqref{ab:Ut}, we obtain the rough bound on $U$
$$
|U| \le C \ln R.
$$
Applying now the commutator estimate from Lemma \ref{601}, we have that
\begin{eqnarray}\label{ineq:uas}
\left| \frac{d}{dR}\aver{U}- \aver{U_R} \right| \le  \frac{ \pi K^2}{R^3} \, \aver{|U|} \big\| \big( e^{2\eta} \big)_\theta \big\|_{L^1(S^1)}.
\end{eqnarray}
From the above rough bound on $U$, we have 
$$| \aver{|U|}| \le C \ln R.$$
Moreover, one can estimate  $\big\| \big( e^{2\eta} \big)_\theta \big\|_{L^1(S^1)}$ as before, to get
$$\big\| \big( e^{2\eta} \big)_\theta \big\|_{L^1(S^1)} \le C R^{3/2}.$$

Thus the right-hand side of \eqref{ineq:uas} is integrable in $R$. 
Since, moreover, 
$$
\aver{U_R}=\frac{1}{\Pscr} \int_0^{2\pi} U_R a^{-1}(R,\theta) d\theta= \frac{R_0}{\Pscr R} \int_0^{2\pi} U_R a^{-1}(R_0,\theta) d\theta,
$$
using the conservation law in Lemma \ref{eq:mvar}, it follows that $\aver{U_R}$ and, therefore, $\frac{d}{dR}\aver{U}$ are integrable. By having checked the convergence of all the integrals involved in our analysis, this 
completes the proof of \eqref{ab:U} and, thus, of Theorem \ref{th:mt}. 
\end{proof}


\section{Future geodesic completeness} \label{se:gc}

In this section, we complete the proof of the geodesic completeness property under the smallness assumption \eqref{ineq:gah}-\eqref{ida:g}. There are only small modifications in comparison to the proof already presented by the authors in \cite{LeFlochSmulevici2} for weakly regular Gowdy spacetimes. One of difficulties (observed and solved in \cite{LeFlochSmulevici2}) is that, with limited control of the Christoffel symbols in the 
$L^1$ or $L^2$ norms (in space) only, the local existence of geodesics is not guaranteed by the
 standard Cauchy-Lipschitz theorem. Instead, we first established that the the Christoffel symbols 
admit traces along timelike curves,
and we relied on a compactness argument `a la Arzela-Ascoli in order to establish the existence of 
geodesics. This part of the analysis can be repeated here almost identically in our $T^2$ setting, 
by using the estimates in \cite{LeFlochSmulevici1} for the compactness argument. (This compactness is required
in the proof of existence of traces , as explained in Proposition 3.5 of \cite{LeFlochSmulevici2}).
We do not repeat these arguments here and directly assume the existence of geodesics (which, for instance, is 
immediate in the smooth case). 

\begin{theorem}[Future geodesic completeness]
 \label{th:gc}
Let $(\mathcal{M},g)$ be a non-flat, polarized $T^2$-symmetric vacuum spacetime with weak regularity 
whose initial data set satisfies the conditions \eqref{ineq:gah}-\eqref{ida:g}. 
Then, all future timelike geodesics are future complete. 
\end{theorem}

\begin{proof} 
For simplicity in the presentation, we focus on the smooth case. 
Let $\xi$ be a future maximal timelike geodesic defined on an interval $[s_0,s_1)$. We have $g(\dot \xi, \dot \xi) < 0$ and
\be
\ddot \xi^\alpha+\Gamma^\alpha_{\beta \gamma} \dot \xi^\beta \dot \xi^\gamma=0.
\ee
Following \cite{LeFlochSmulevici2}, we observe that, since $X$ and $Y$ are Killing fields, $J_X=g(\dot \xi,X)$ and $J_Y=g(\dot \xi,Y)$ are constant along $\xi$, so that 
$J_X=e^{2U}\left( \dot \xi^X+ G \dot \xi^\theta \right)$ and $J_Y=e^{-2U} R^2 \left( \dot \xi^Y+ H \dot \xi^\theta \right)$ are constants along $\xi$.
We  use the same strategy as in Section 4 of \cite{LeFlochSmulevici2}. First, by standard arguments (see Lemma 4.10 in \cite{LeFlochSmulevici2}), it follows that $R(\xi(s)) \rightarrow +\infty$ as $s \rightarrow s_1$. Then, since $R(\xi(s)) -R(\xi(s_0)) = \int_{s_0}^{s} \dot \xi^R ds$, it follows that any bound of the form $\dot \xi^R < C R^p$ for $p < 1$ implies that $s_1=+\infty$. Note also that since $R(\xi(s)) \rightarrow +\infty$, given any $R' > 0$, we may assume, without loss of generality, that $R(\xi(s_0)) \ge R'$.

We now analyze the structure of the equation satisfied by $\dot \xi^R$: 
\begin{equation}
\label{eq:xirdd}
\ddot \xi^R+\Gamma^R_{\beta \gamma} \dot \xi^\beta \dot \xi^\gamma=0.
\end{equation}
The term $\Gamma^R_{\beta \gamma} \dot \xi^\beta \dot \xi^\gamma=0$ is decomposed in the form 
$$
\Gamma^R_{\beta \gamma} \dot \xi^\beta \dot \xi^\gamma=\Gamma^R_{RR} \dot \xi^R \dot \xi^R+ \Gamma^R_{\theta\theta} \dot \xi^\theta \dot \xi^\theta+ 2\Gamma^R_{R\theta} \dot \xi^R \dot \xi^\theta+ 2\Gamma^R_{\theta a} \dot \xi^\theta \dot \xi^a +\Gamma^R_{a b}\dot \xi^a \dot\xi^b,
$$
where $a,b=X,Y$.
Recall now that 
\begin{eqnarray}
\Gamma^R_{RR}&=&\eta_R-U_R, \\
\Gamma^R_{\theta\theta}&=&\frac{\eta_R-U_R}{a^2}-\frac{a_R}{a^3}+e^{2U}U_R G^2 e^{-2(\eta-U)}+\left( e^{-2U}R^2 H^2 \right)_R \frac{e^{-2(\eta-U)}}{2}, \\
\Gamma^R_{R\theta}&=&\eta_\theta-U_\theta.
\end{eqnarray}
Observe also that 
$$
\eta_R-U_R= R \left( \left(U_R-\frac{1}{2R}\right)^2+a^2 U_\theta^2 \right)-\frac{1}{4R}-\frac{K^2}{4R^3}e^{2\eta},
$$
while 
$$
\eta_\theta-U_\theta=2R \left(U_R-\frac{1}{2R}\right) U_\theta.
$$
As a consequence, it follows that the following quadratic form inequality holds
\be
(\eta_R-U_R) (dR^2 +a^{-2}d\theta^2)+2 \left(\eta_\theta-U_\theta\right) dR d\theta+ \left( \frac{1}{4R}+\frac{K^2}{4R^3}e^{2\eta} \right) \left(dR^2+a^{-2}d\theta^2\right)\ge 0.
\ee

Returning now to \eqref{eq:xirdd}, this leads us to 
$$
\aligned
\ddot \xi^R 
& \le  \left( \frac{1}{4R}+\frac{K^2}{4R^3}e^{2\eta} \right) \left((\dot \xi^R)^2+a^{-2}(\dot \xi^\theta)^2\right)+\frac{a_R}{a^3}(\dot \xi^\theta)^2  
\\
& \quad -\left( e^{2U}U_R G^2 e^{-2(\eta-U)}+\left( e^{-2U}R^2 H^2 \right)_R \frac{e^{-2(\eta-U)}}{2}\right)(\dot \xi^\theta)^2  
-2\Gamma^R_{\theta a} \dot \xi^\theta \dot \xi^a -\Gamma^R_{a b}\dot \xi^a\dot \xi^b.
\endaligned
$$
Note that the term containing $\frac{a_R}{a^3}$ has the right-sign and can absorb the term $\frac{K^2}{4R^3}e^{2\eta}\dot (\xi^\theta)^2$. Using moreover the estimate \eqref{ab:et} and the fact that $|a^{-1}\dot \xi^\theta| \le  \dot\xi^R$, for all $\epsilon > 0$, we may assume that $R(\xi(s_0)$ is sufficiently large so that
$$
\ddot \xi^R \le  \left( \frac{3+\epsilon}{4 R} \right) (\dot \xi^R)^2-2\Gamma^R_{\theta a} \dot \xi^\theta \dot \xi^a -\Gamma^R_{a b} \dot \xi^a \dot \xi^b-\left( e^{2U}U_R G^2 e^{-2(\eta-U)}+\left( e^{-2U}R^2 H^2 \right)_R \frac{e^{-2(\eta-U)}}{2}\right)(\dot \xi^\theta)^2  .
$$

Recalling now that $\frac{d R}{ds}=\dot \xi^R$, the last inequality can be rewritten as
\be
\label{es:xir}
\aligned
& \frac{d}{ds}\left( R^{-3/4-\epsilon} \dot \xi^R \right) 
\\
& \le R^{-3/4-\epsilon} \left(-\left( e^{2U}U_R G^2 e^{-2(\eta-U)}+\left( e^{-2U}R^2 H^2 \right)_R \frac{e^{-2(\eta-U)}}{2}\right)(\dot \xi^\theta)^2 - 2\Gamma^R_{\theta a} \dot \xi^\theta \dot \xi^a -\Gamma^R_{a b} \dot \xi^a \dot \xi^b\right).
\endaligned
\ee
For the three terms in the right-hand side, recall that
\begin{eqnarray*}
\Gamma_{X\theta}^R&=&e^{-2\eta}e^{4U} U_R G,
 \\
\Gamma_{Y\theta}^R&=&\frac{e^{-2(\eta-U)}}{2} \left( e^{-2U} R^2 H \right)_R, 
\\
\Gamma_{XX}^R&=& e^{-2\eta}e^{4U} U_R,
 \\
\Gamma_{XY}^R&=&0,
 \\
\Gamma_{YY}^R&=&\frac{e^{-2(\eta-U)}}{2} \left( e^{-2U} R^2 \right)_{R}. 
\end{eqnarray*}
These terms can be combined with the terms containing $H^2$ and $G^2$ above arising from $\Gamma_{\theta \theta}^R$ as follows: 
\begin{eqnarray*}
\Gamma^{R}_{XX} (\dot{\xi}^X)^2+2\Gamma^{R}_{\theta X} \dot{\xi}^\theta \dot{\xi}^X + e^{2U} U_R G^2 e^{-2(\eta-U)} \left(  \dot{\xi}^\theta\right)^2&=&e^{-2(\eta-U)} e^{2U} U_R \left( \dot{\xi}^X + G \dot{\xi}^\theta \right)^2 \\
&=&e^{-2\eta} U_R \, J_X^2
\end{eqnarray*}
and 
$$
\aligned
& \Gamma^{R}_{YY} (\dot{\xi}^Y)^2+2\Gamma^{R}_{\theta Y} \dot{\xi}^\theta \dot{\xi}^Y +\left( e^{-2U}R^2 H^2 \right)_R \frac{e^{-2(\eta-U)}}{2} \left(  \dot{\xi}^\theta\right)^2
\\
& = \frac{e^{-2(\eta-U)}}{2} \left( \left( e^{-2U}R^2 \right)_R \left( \dot{\xi}^Y+ H \dot{\xi}^\theta \right)^2+e^{-2U}R^2 2 H H_R   \left(  \dot{\xi}^\theta\right)^2  +2 e^{-2U} H_R R^2 \dot \xi^\theta \dot \xi^Y \right)\\
\\
&= \frac{e^{-2(\eta-U)}}{2} \left(\left( e^{-2U}R^2 \right)_R R^{-4} e^{4U}\, J_Y^2+ 2 H_R \dot \xi^\theta J^Y \right).
\endaligned
$$

Now let $\mu=\eta-U+\frac{1}{4} \ln R-\frac{1}{2}\ln a$ and note that 
$$
\mu_R=R \left( \left(U_R-\frac{1}{2R}\right)^2+a^2 U_\theta^2 \right) \ge 0.
$$
Then, using that $U$ is uniformly bounded and \eqref{ab:et}, we easily have the estimates
\begin{eqnarray}
\left| e^{-2\eta} U_R \, J_X^2 \right| &\le&C R^{-2}\left( R^{-1/2} \mu_R^{1/2} +\frac{1}{R} \right),
  \\
\left|\frac{e^{-2(\eta-U)}}{2} \left( e^{-2U}R^2 \right)_R R^{-4} e^{4U}\, J_Y^2 \right| &\le& C R^{-4}\left( R^{-1/2} \mu_R^{1/2} +\frac{1}{R} \right)  
\end{eqnarray}
for some constant $C > 0$.
Moreover, in view of equation \eqref{GHequa}, \eqref{ab:et} and the estimate $|\dot \xi^\theta| \le a \dot \xi^R$, 
$$
\left| e^{-2(\eta-U)} H_R \dot \xi^\theta J^Y \right| \le C \frac{\dot \xi^{R} }{R^3}.
$$
Returning to \eqref{es:xir}, we obtain
$$
\frac{d}{ds}\left( R^{-3/4-\epsilon} \dot \xi^R \right) \le C R^{-13/4-\epsilon} \left(\mu_R^{1/2}+R^{-1/2}\right)+C \frac{\dot \xi^{R} }{R^3} .
$$
The second term in the right-hand side is integrable since $\dot \xi^{R}=\frac{dR(\xi(s)}{ds}$. Moreover, $R^{-13/4-\epsilon}R^{-1/2}$ is decreasing in $R$ and, therefore, integrable on any bounded interval $[s_0,s_1]$. Thus, it remains only to show that $R^{-13/4-\epsilon} \mu_R^{1/2}$ is integrable.

Let $M^2=-g(\dot \xi,\dot \xi)$. Then, we have 
$$
a^{-2} \left(\frac{\dot \xi^\theta}{\dot \xi^R} \right)^2 \le 1- \frac{M^2e^{-2(\eta-U)}}{\left(\dot \xi^R\right)^2}.
$$
Let $\chi=\frac{M^2e^{-2(\eta-U)}}{\left(\dot \xi^R\right)^2} \le 1$ and let $\rho=\eta-U$. Then, we find\footnote{We would like here to consider $\frac{d \mu}{ds}$, however, this would introduce the quantity $a_\theta$ for which we do not directly have an evolution equation.}
\begin{eqnarray}
\frac{d \rho}{ds} +1/4\frac{d}{ds}\left( \ln R \right) - \frac{a_R}{2a} \dot \xi^R &\ge& \left(1-(1-\chi)^{1/2}\right) \mu_R \dot \xi^R 
\\
&\ge&  1/2 \chi \mu_R \dot \xi^R.
\end{eqnarray}
In particular, $\frac{d \rho}{ds} +1/4\frac{d}{ds}\left( \ln R \right) - \frac{a_R}{2a} \dot \xi^R \ge 0$.
As a consequence, we have
$$
\mu_R \le 2 \left( \frac{d \rho}{ds} +1/4\frac{d}{ds}\left( \ln R \right) - \frac{a_R}{2a} \dot \xi^R \right)M^{-2} e^{2\rho} \dot \xi^R.
$$
Now, recall that from \eqref{ab:et}

$$
- \frac{a_R}{2a} = \frac{1}{4 R} + O( R^{-5/4}).
$$
In particular, there exists some $R_2 > 0$, such that for all $s$ such that $R(\xi(s)) > R_2$, 
$$- \frac{a_R}{2a} \le \frac{1+\epsilon}{4 R},$$
and we can assume that $R(\xi(s_0)) \ge R_2$.
Thus, we have
$$
\mu_R \le 2 \left( \frac{d \rho}{ds} +\frac{1+\frac{\epsilon}{2}}{2}\frac{d}{ds}\left( \ln R \right)  \right)M^{-2} e^{2\rho} \dot \xi^R,
$$
 where the quantity in the parentheses $\frac{d \rho}{ds} +\frac{1+\frac{\epsilon}{2}}{2}\frac{d}{ds}\left( \ln R \right) \ge 0$ is positive. 

Thus, we conclude that 
\begin{eqnarray*}
\mu_R^{1/2} &\le& \sqrt{2}M^{-1} \left( \frac{d \rho}{ds} +\frac{1+\frac{\epsilon}{2}}{2}\frac{d}{ds}\left( \ln R \right) \right)^{1/2} e^{\rho} \left( \dot \xi^R\right)^{1/2}, \\
&\le&
C \left( \frac{d \rho}{ds} +\frac{1+\frac{\epsilon}{2}}{2}\frac{d}{ds}\left( \ln R \right)  \right)e^{2\rho}+  C \dot \xi^R.
\end{eqnarray*}
It follows that 
$$
R^{-13/4-\epsilon} \mu_R^{1/2} \le C R^{-13/4-\epsilon}\left( \frac{d \rho}{ds} +\frac{1+\frac{\epsilon}{2}}{2}\frac{d}{ds}\left( \ln R \right)  \right)e^{2\rho} +C R^{-13/4-\epsilon}\dot \xi^R,
$$
where the last term is clearly integrable since $\dot \xi^R=\frac{dR(\xi(s))}{ds}$ and $13/4-\epsilon > 1$. Finally, using  \eqref{ab:et} and an integration by parts to estimate the term containing $\frac{d \rho}{ds}$, we have, for any $s \in [s_0,s_1)$
\begin{eqnarray*}
&&\int_{s_0}^{s} R^{-13/4-\epsilon}\left( \frac{d \rho}{ds} +\frac{1+\frac{\epsilon}{2}}{2}\frac{d}{ds}\left( \ln R \right)  \right)e^{2\rho} ds
\\
&&= \int_{s_0}^{s} R^{-13/4-\epsilon}1/2\frac{d e^{2\rho}}{ds}ds  +\int_{s_0}^{s} R^{-13/4-\epsilon}\frac{1+\frac{\epsilon}{2}}{2}\frac{d}{ds}\left( \ln R \right) e^{2\rho} ds
\\
&&\le C e^{2\rho} R^{-13/4-\epsilon} + C\int_{s_0}^{s} R^{-17/4-\epsilon} \dot \xi^R e^{2\rho} ds +C\int_{s_0}^{s} R^{-9/4} \dot \xi^R ds
\le C.
\end{eqnarray*} 

Thus, we have shown that $\frac{d}{ds} \left( R^{-3/4-\epsilon}\dot \xi^R \right)$ is integrable and, therefore, that $\dot \xi^R \le C R^{3/4+\epsilon}$, for some $C > 0$. This completes the proof of Theorem \ref{th:gc}. 
\end{proof}


\section{Existence of initial data sets close to the asymptotic regime}

In this section, we prove the following result. 
 
\begin{proposition}[Existence of a class of initial data sets] 
Fix $C_1> 0$ and $\mathcal{A} \in [0,+\infty)$. For any $\epsilon > 0$, there exists $R_0 > 0$, $(U_0,U_1) \in H^1(S^1) \times L^2(S^1), a_0 > 0 \in W^{2,1}(S^1)$ and $\eta_0 \in W^{1,1}(S^1)$ such that $\left(U_0,U_1, a_0,\eta_0\right)$ satisfies the constraint equation \eqref{weakconstraintsr2}, that is, 
\begin{eqnarray} \label{cs:data}
\partial_\theta(\eta_0)  = 2 R_0 \,  U_1 \partial_\theta(U_0)
\end{eqnarray}
and such that the conditions \eqref{ineq:gah}-\eqref{sm:g} are all satisfied with $U(R_0,\theta)=U_0(\theta)$, $U_R(R_0,\theta)=U_1(\theta)$, $\eta(R_0,\theta)=\eta_0(\theta)$ and $a(R_0,\theta)=a_0(\theta)$. 
As a consequence, there exists an non-empty set of initial data satisfying \eqref{ineq:gah}-\eqref{sm:g} which is open in the natural topology associated with the initial data on $H^1(S^1) \times L^2(S^1) \times W^{2,1}(S^1) \times W^{1,1}(S^1)$. 
\end{proposition}

While our construction require us to choose a sufficiently large $R_0$ (depending on $\epsilon$), the $\epsilon$ satisfying the assumption of Theorem \eqref{th:mt} depends only on a lower bound on $R_0$. Hence, the data constructed above satisfy the requirements of Theorem \eqref{th:mt} provided $R_0$ is chosen sufficiently large.

\begin{proof}
Let $C_1 > 0$ and $\mathcal{A} \in [0,+\infty)$ be fixed. We define $a_0$ to be 
$$
a_0= \frac{2\pi}{p R_0^{1/2}},
$$ 
where $p > 0$ is a constant. Thus the associated term $\Pscr$ reads $\Pscr=p R_0^{1/2}$. We then define $U_1$ as 
$$
U_1 = \pm\frac{\mathcal{A}^{1/2}}{p R_0^{3/2}}, 
$$
so that
$$
\left( R_0 \int_0^{2\pi} U_1 a_0^{-1} d\theta \right)^{2}=\mathcal{A}.
$$
Consider now any non-constant $U_0 \in H^1(S^1)$. We will impose several conditions on $U_0$.

Let $\Ecal=\int_{S^1} \left( a_0^{-1}U_1^2+a_0 (U_0)_\theta^2 \right) d \theta$ be the energy associated with our initial data set. Note that the energy correction\footnote{We would like to thank an anonymous referee for pointing out this nice simplification.} $\Gamma^U=\frac{1}{R_0}\int_{S^1}\left( U_0 -<U_0> \right) U_1 a_0^{-1} dR=0$ since $U_1a_0^{-1}$ is constant. 

Let $\mathcal{F}=\Pscr \Escr$ and $\mathcal{G}=\Pscr \left( \Escr+\Gamma^U \right) $ be the rescaled energy and the rescaled corrected energy associated with $U_0$, $U_1$ and $a_0$. Note that since $\Gamma^U=0$, $\mathcal{G}=\mathcal{F}$, so that \eqref{es:fg} trivially holds. 
Observe that 
$$ 
\Fcal= \Pscr \int_0^{2\pi}\left( a_0^{-1} U_1^2+ a_0 (U_0)_\theta^2\right) d\theta
= \frac{\mathcal{A}}{R_0^2} +2\pi \int_0^{2\pi} (U_0)_\theta^2 d\theta.
$$

Suppose now that $\int_0^{2\pi} (U_0)_\theta^2 d\theta= \frac{ C_1}{2 \pi R_0},$ where $C_1 > 0$.
Then, we have 
$$
\Fcal = \frac{\mathcal{A}}{R_0^2} + \frac{C_1}{R_0}. 
$$
In order to satisfy \eqref{ida:c}, we now fix $p$ in terms of $C_1$ by setting 
$$
p= \left(\frac{2C_1}{5}\right)^{1/2}.
$$
Then, we compute 
$$
|c_1|=\left| \frac{2}{\sqrt{5}} - \frac{\Pscr}{R_0 \Gcal^{1/2}} \right|=\left| \frac{2}{\sqrt{5}} -\frac{ \Pscr}{R_0 \Fcal^{1/2}} \right|.
$$ 
On the other hand, we have  
$$
\frac{\Pscr}{R_0 \Fcal^{1/2}}= \frac{p R_0^{1/2}}{R_0\left(\frac{\mathcal{A}}{R_0^2} + \frac{C_1}{R_0}\right)^{1/2}}=\frac{2}{\sqrt{5}} \left( 1 + \frac{\mathcal{A}}{ R_0 C_1 } \right)^{-1}.
$$
This shows that \eqref{ida:c} is satisfied provided that $\frac{\mathcal{A}}{ R_0 C_1 }$ is sufficiently small, which we can always ensure by choosing $R_0$ sufficiently large compared to $\frac{\mathcal{A}}{C_1}$.
 
One can then easily check that \eqref{ineq:gah} and \eqref{ida:g} hold true provided $R_0$ is sufficiently large. 
It remains to define $\eta_0$ so that \eqref{ida:d} and the constraint equation \eqref{cs:data} is satisfied. 
 
For \eqref{ida:d}, we only need to ensure that 
$\left| \frac{\Qcal}{R_0^3 \Fcal^{1/2}}-\frac{1}{\sqrt{5}} \right| \le \epsilon.$
Recall that $\Qcal=\int_{0}^{2\pi} \frac{K^2}{2} e^{2\eta_0} a_0^{-1} d\theta$.
By fixing $\eta_0(\theta=0)$ we can certainly ensure that 
$$
\frac{K^2}{2} e^{2\eta_0(0)} a_0^{-1}=\frac{1}{2 \pi \sqrt{5}} R_0^{3} \Fcal^{1/2}.
$$
Now we define $\eta_0$ for all other values of $\theta$ so that \eqref{cs:data} is satisfied
\begin{eqnarray*}
\eta_0(\theta)&=& \eta_0(0)+ 2 R_0 \int_0^\theta U_1 (U_0)_\theta d\theta,\\
&=& \eta_0(0)+ 2 R_0 U_1\left( U_0(\theta)-U_0(0)\right).
\end{eqnarray*}
From the above, we see that $\eta_0 \in W^{1,1}(S^1)$ (and in fact in $H^1(S^1)$)  and that
\begin{eqnarray*}
| \eta_0(0)-\eta_0(\theta) | &\le&\int_0^{2\pi} | \eta_\theta |\, d\theta \le  R_0 \frac{\Fcal}{\Pscr}
\\
 &\le&  \frac{1}{p R_0^{1/2}} R_0 \left( \frac{\mathcal{A}}{R_0^2}+ \frac{C_1}{R_0} \right)
\le \epsilon^2, 
\end{eqnarray*} 
again by choosing $R_0$ sufficiently large depending only on $C_1$ and $\mathcal{A}$.
We then check that 
$$
\aligned
\left| \frac{\Qcal}{R_0^3 \Fcal^{1/2}} - \frac{1}{\sqrt{5}} \right| 
&= \frac{1}{R_0^3\Fcal^{1/2}}\left| \Qcal - 2 \pi \frac{K^2}{2} e^{2\eta_0(0)}a_0^{-1} \right| 
\\
&\le  \frac{1}{R_0^3\Fcal^{1/2}}\frac{K^2}{2} e^{2\eta_0(0)} a_0^{-1} \int_0^{2\pi} | e^{2(\eta_0(\theta)-\eta_0(0))}-1 |\, d\theta 
 \le  C \epsilon^2 \le \epsilon, 
\endaligned
$$
provided $\epsilon$ is sufficiently small. 
\end{proof}


\section*{Acknowledgments}

The authors are very grateful to Hans Ringstr\"om for stimulating discussions on this subject and useful comments on a first version of this paper, and are also thankful to an anonymous referee for many constructive remarks.
Part of this paper was written in the Fall Semester 2013 when the first author (PLF) was a visiting professor at the Mathematical Sciences Research Institute (Berkeley) thanks to the support of the 
National Science Foundation under Grant No. 0932078 000. 
Part of this paper was written while the second author (JS) was a member of the 
Max Planck Institute for Gravitational Physics (Albert Einstein Institute). 
The authors were also partially supported by the Agence Nationale de la Recherche through the grant ANR SIMI-1-003-01.



\small 

\begin{thebibliography}{10} 

\bibitem{BergerChruscielMoncrief} \auth{Berger B.K., Chru\'sciel P., and Moncrief V.,}
On asymptotically flat spacetimes with $G_2$-invariant Cauchy surfaces, 
Ann. Phys. 237 (1995), 322--354. 

\bibitem{BergerChruscielIsenbergMoncrief} 
\auth{Berger B.K., Chru\'sciel P., Isenberg J., and Moncrief V.,}
Global foliations of vacuum spacetimes with $T^2$ isometry, 
Ann. Phys. 260 (1997), 117--148. 

\bibitem{ChoquetBruhat} 
\auth{Choquet-Bruhat Y.,}
Future complete $S^1$ symmetric Einsteinian spacetimes, the unpolarized case,
C. R. Math. Acad. Sci. Paris 337 (2003), 129--136. 

\bibitem{ChoquetBruhatMoncrief}
\auth{Choquet-Bruhat Y. and Moncrief V.,}
Future global in time Einsteinian spacetimes with $\rm U(1)$ isometry group,
Ann. Henri Poincar\'e 2 (2001), 1007--1064. 

\bibitem{Christodoulou} \auth{Christodoulou D.,} 
Bounded variation solutions of the spherically symmetric Einstein-scalar field equations, 
Comm. Pure Appl. Math. 46 (1992), 1131--1220.  

\bibitem{Chrusciel} \auth{Chru\'sciel P.,} 
On spacetimes with $U(1) \times U(1)$ symmetric compact Cauchy surfaces,
Ann. Phys. 202 (1990), 100--150.

\bibitem{ChruscielIsenbergMoncrief} \auth{Chru\'sciel P., Isenberg J., and Moncrief V.,} 
Strong cosmic censorship in polarized Gowdy spacetimes, 
Class. Quantum Grav. 7 (1990), 1671--1680. 

\bibitem{EardleyMoncrief} \auth{Eardley D. and Moncrief V.,}
The global existence problem and cosmic censorship in general relativity, 
Gen. Relat. Grav. 13 (1981), 887--892.  

\bibitem{Gowdy} \auth{Gowdy R.,} 
Vacuum spacetimes with two-parameter spacelike isometry groups and compact invariant hypersurfaces: 
topologies and boundary conditions, 
Ann. Phys. 83 (1974), 203--241.  

\bibitem{GL1} N. Grubic and P.G. LeFloch, 
Weakly regular Einstein--Euler spacetimes with Gowdy symmetry. The global areal foliation, 
Arch. Rational Mech. Anal. 208 (2013), 391--428.

\bibitem{GL2} N. Grubic and P.G. LeFloch,   
On the area of the symmetry orbits in weakly regular Einstein-Euler spacetimes with Gowdy symmetry, 
SIAM J. Math. Anal. 47 (2015), 669--683.

\bibitem{IsenbergMoncrief} \auth{Isenberg J. and Moncrief V.,}
Asymptotic behavior of the gravitational field and the nature of singularities in Gowdy spacetimes, 
Ann. Phys. 99 (1990), 84--122. 

\bibitem{IsenbergWeaver} \auth{Isenberg J. and Weaver M.,}
On the area of the symmetry orbits in $T^2$--symmetric spacetimes, 
Class. Quantum Grav. 20 (2003), 3783--3796.  

\bibitem{LeFloch} \auth{LeFloch P.G.,}
An introduction to self-gravitating matter, Graduate Course given at the Institute Henri Poincar\'e, Paris, Fall 2015. Videos available at  https://www.youtube.com/user/PoincareInstitute 

\bibitem{LeFlochMardare} \auth{LeFloch P.G. and Mardare C.,}
Definition and weak stability of spacetimes with distributional curvature, 
Port. Math. 64 (2007), 535--573.

\bibitem{LeFlochRendall} \auth{LeFloch P.G. and Rendall A.,}
A global foliation of Einstein-Euler spacetimes with Gowdy-symmetry on $T^3$, 
Arch. Rational Mech. Anal. 201 (2011), 841--870.   

\bibitem{LeFlochSmulevici1} \auth{LeFloch P.G. and Smulevici J.,}
Weakly regular $T^2$--symmetric spacetimes. 
The global geometry of future Cauchy developments,  
J. Eur. Math. Soc. 17 (2015), 1229--1292.

\bibitem{LeFlochSmulevici2} \auth{LeFloch P.G. and Smulevici J.,}
Weakly regular T2-symmetric spacetimes. The future causal geometry of Gowdy spaces, 
J. Differ. Equa. 260 (2016), 1496--1521.

\bibitem{LeFlochSormani} \auth{LeFloch P.G. and Sormani C.,} 
The nonlinear stability of rotationally symmetric spaces with low regularity,
 J. Funct. Anal. 268 (2015), 2005--2065. 
 
\bibitem{LeFlochStewart} \auth{LeFloch P.G. and Stewart J.M.,}
Shock waves and gravitational waves in matter spacetimes with Gowdy symmetry, 
Port. Math. 62 (2005), 349--370. 

\bibitem{LeFlochStewart2} \auth{LeFloch P.G. and Stewart J.M.,}
The characteristic initial value problem for plane--symmetric spacetimes with 
weak regularity, Class. Quantum Grav. 28 (2011), 145019--145035.  

\bibitem{Moncrief} \auth{Moncrief V.,}
Global properties of Gowdy spacetimes with $T^3 \times \RR$ topology, 
Ann. Phys. 132 (1981), 87--107. 

\bibitem{Rendall-crush} \auth{Rendall A.D.,} 
Crushing singularities in spacetimes with spherical, plane, and hyperbolic symmetry, 
Class. Quantum Grav. 12 (1995), 1517--1533. 

\bibitem{Rendall2} \auth{Rendall A.D.,} 
Existence of constant mean curvature foliations in spacetimes with two-dimensional local symmetry,
Commun. Math. Phys. 189 (1997), 145--164.

\bibitem{Ringstrom} \auth{Ringstr\"om H.,}
On a wave map equation arising in general relativity,
Comm. Pure Appl. Math. 57  (2004), 657--703.

\bibitem{Ringstrom1}\auth{Ringstr\"om H.,} 
Curvature blow-up on a dense subset of the singularity in $T^3$-Gowdy, 
Jour. Hyper. Diff. Equa. 2 (2005), 547--564. 

\bibitem{Ringstrom2}\auth{Ringstr\"om H.,} 
Strong cosmic censorship in $T^3$-Gowdy spacetimes, 
Ann. Math. 170 (2009), 1181--1240. 

\bibitem{Ringstrom3}\auth{Ringstr\"om H.,} 
Instability of spatially homogeneous solutions in the class of $T^2$-symmetric solutions to Einstein's vacuum equations. 

\bibitem{Smulevici1} \auth{Smulevici J.,}
Strong cosmic censorship for $T^2$--symmetric spacetimes with positive cosmological constant and matter, 
Ann. Henri Poincar\'e 9 (2009), 1425--1453.

\bibitem{Smulevici2} \auth{Smulevici J.,}
On the area of the symmetry orbits in spacetimes with toroidal or hyperbolic symmetry, 
Anal. \& PDE 4 (2011), 191--245.

\end{thebibliography}
\end{document}